\documentclass[11pt]{article}
\usepackage{ogp}
\usepackage{setspace}
\usepackage{cprotect}
\usepackage{makeidx}
\usepackage[mathscr]{euscript}
\usepackage{accents}

\def\anon{0}
\def\authornotes{0}

\ifnum\authornotes=1
    \newcommand{\shuangping}[1]{\footnote{\color{purple}Shuangping: {#1}}}
    \newcommand{\tselil}[1]{\footnote{\color{ForestGreen}Tselil: {#1}}}

    \newcommand{\tnote}[1]{{\color{ForestGreen}[Tselil: #1]}}
    \newcommand{\snote}[1]{{\color{Purple}[Shuangping: #1]}}

	\newcommand{\todo}[1]{{\color{Red}[TODO: #1]}}	
\else
    \newcommand{\shuangping}[1]{}
    \newcommand{\tselil}[1]{}

    \newcommand{\tnote}[1]{}
    \newcommand{\snote}[1]{}

	\newcommand{\todo}[1]{}
\fi

\title{Some easy optimization problems have the overlap-gap property}

\ifnum\anon=0
\author{Shuangping Li\thanks{Stanford University. \texttt{fifalsp@stanford.edu}.} \and Tselil Schramm\thanks{Stanford University.  \texttt{tselil@stanford.edu}. Supported by NSF CAREER award \# 2143246.}}
\else
\author{Author name(s) withheld for double-blind review}
\fi

\date{}

\begin{document}
\maketitle

\begin{abstract}
We show that the shortest $s$-$t$ path problem has the overlap-gap property in (i) sparse $\bbG(n,p)$ graphs and (ii) complete graphs with i.i.d. Exponential edge weights.
Furthermore, we demonstrate that in sparse $\bbG(n,p)$ graphs, shortest path is solved by $O(\log n)$-degree polynomial estimators, and a uniform approximate shortest path can be sampled in polynomial time.
This constitutes the first example in which the overlap-gap property is not predictive of algorithmic intractability for a (non-algebraic) average-case optimization problem.
\end{abstract}

{ \hypersetup{hidelinks} \tableofcontents }
\thispagestyle{empty}
\clearpage
\setcounter{page}{1}
\section{Introduction}
An instance $\calI$ of an optimization problem is said to have the \emph{overlap-gap property} (OGP) if its near-optimal solutions form multiple well-separated clusters.
That is, every pair of near-optimal solutions $S_1,S_2$ are either close, with  $\dist(S_1,S_2) \le \eps$ (in some distance metric over solutions), or separated, with $\dist(S_1,S_2) \ge 1- \delta$, and furthermore $\calI$ has at least one pair of far apart solutions.

The definition was inspired by complexity-theoretic heuristics in statistical physics, developed in the context of random CSPs such as $k$-SAT \cite{MMZ05,ART06}.
Suppose you have an interacting particle system whose low-energy states are the near-optimal solutions to the optimization problem $\calI$. 
If $\calI$ has low-energy states $S_1,S_2$ that are far apart, then in equilibrium the system should occasionally transition between $S_1,S_2$; if the solutions to $\calI$ are clustered according to the OGP, then traversing from $S_1$ to $S_2$ requires passing through high-energy states, which should be hard for ``natural'' algorithms.

This heuristic picture can be translated into \emph{unconditional} lower bounds against specific algorithms, particularly in the context of average-case optimization where one can precisely characterize the optimization landscape.
Gamarnik and Sudan \cite{GS17} were the first to show that the OGP implies that local algorithms cannot find large independent sets in sparse random graphs.
This simple and powerful idea has since been extended, and variations of the OGP are known to imply unconditional lower bounds against sufficiently ``smooth'' or Lipschitz algorithms such as Langevin Dynamics and Approximate Message Passing \cite{GJW24,HS22}, low-degree polynomial algorithms \cite{GJW24,Wein22}, and generally any algorithm which is stable to perturbations in the input.
A detailed explanation and further references are available in the survey of Gamarnik \cite{Gam21}.
A strengthening of the property known as \emph{disorder chaos} has recently been used to give unconditional lower bounds for sampling from Ising models \cite{el2022sampling}.

Remarkably, in many average-case optimization problems, the overlap-gap property kicks in at the known computational threshold.
For example, consider the maximum independent set problem in a random graph of average degree $d$.
The size of the maximum independent set concentrates very well around $2\log d/d$.
Despite decades of effort, the best known algorithms can reliably find a set half this size, but no larger.
Strikingly, a sequence of works \cite{GS17,RV17,GJW24,Wein22} has shown that a type of overlap-gap property holds for independent sets of measure $(1+\eps)\log d/d$, from which it is possible to conclude that no local/smooth/low-degree algorithm can reliably find independent sets of measure $>\log d/d$ (which is precisely where present day polynomial-time algorithms are stuck).
The OGP is similarly consistent with computational limits for random $k$-SAT \cite{BH22}, random Ising models and other spin glasses \cite{GJW24,HS22,HS23,el2022sampling}, random knapsack \cite{GK23}, and more.
The only example where OGP is known to give an incorrect prediction of hardness is the $k$-XOR problem, which has the OGP even while it is easy to solve via Gaussian elimination \cite{IKKM12}.
The algebraic structure of $k$-XOR is known to throw off heuristics for predicting average-case complexity, for example for the planted analogue of $k$-XOR or for the related problem of ``Learning Parity with Noise,'' so researchers have mostly shrugged this example off.
In the last five years or so, it is the authors' sense that researchers have come to accept the overlap-gap property as plausible evidence for computational intractability, perhaps not just for local/smooth/low-degree algorithms, but for algorithmic methods more generally.

\medskip
The purpose of this paper is to caution against complacency regarding OGP lower bounds.
Our main result is that the algorithmically easy shortest path problem has the overlap gap property in random graphs.
Indeed, we show that the shortest path problem can even be solved by low-degree polynomial algorithms in sparse random graphs, and further, that approximate short paths are easy to sample.

\subsection{Main results}

We study the problem of the shortest $s$-$t$ path in the Erd\"os-R\'enyi graph $\bbG(n,q)$, for $q = \Theta( \frac{\log n}{n})$ above the connectivity threshold.
Polynomial-time algorithms for $s$-$t$ path are a staple of undergraduate algorithms courses.
Yet a simple application of the first moment method demonstrates that this problem has the OGP.

\begin{theorem}[Informal version of \pref{thm:ogp-gnp}]\label{thm:ogp-intro}
The shortest $(s,t)$-path problem in $\bbG(n,C\frac{\log n}{n})$ has the overlap-gap property with high probability when $C > 1$.
Furthermore, if $\bG,\bG' \sim \bbG(n,C\frac{\log n}{n})$ independently, then with high probability all near-shortest $(s,t)$-paths in $\bG$ and $\bG'$ are almost disjoint.
\end{theorem}

We say that an algorithm is ``stable'' or ``smooth'' if it is Lipschitz in its inputs; that is, if $\bG,\bG'$ are two graphs which differ in an $\alpha$ fraction of edges, then algorithm $A$ is $L$-smooth if the output paths $A(\bG),A(\bG')$ differ in at most an $L\cdot \alpha$ fraction of their edges.
This notion of smoothness can be generalized in a straightforward way to apply to randomized algorithms, etc.
Via established techniques, we can show that the OGP implies that no smooth algorithm can reliably solve the $(s,t)$-path problem in $\bbG(n,q)$.
\begin{corollary}[Informal version of \pref{cor:gnp-no-stable}]
Suppose $A$ is a ``stable'' algorithm in the sense that when it runs and succeeds on graphs $G,G'$ which agree on a $\rho$-fraction of edges, it produces paths $A(G),A(G')$ overlapping in a $\Omega(\eps)$-fraction of their edges.
Then $A$ cannot have success probability better than $\frac{1-\rho}{6}$ in computing $(1+\eps)$-approximate $(s,t)$-shortest paths in $\bbG(n,C\frac{\log n}{n})$.
\end{corollary}

One concrete class of algorithms which behave stably on random inputs is \emph{low-degree polynomial algorithms}, that is, algorithms $A$ where $A(G)$ is a vector-valued polynomial, each entry of which is a bounded-degree polynomial in the adjacency matrix of graph $G$. 
Intuitively, when applied to correlated random graphs $\bG,\bG'$, a low-degree polynomial ought to be somewhat stable since low-degree functions are relatively resilient to noise.

This is an interesting class of algorithms because $O(\log n)$-degree polynomials appear to work as well as any other polynomial-time algorithm for a large class of average-case \emph{planted problems}---a good overview is given in the thesis \cite{hopkins2018statistical} and the survey \cite{KWB19}, though there have been many developments since these were written.
Mild stability of bounded-degree polynomials has been established in prior work on OGP; usually the resulting bounds imply that degree-$D$ polynomials must fail with probability $\ge n^{-O(D)}$ \cite{GJW24,Wein22}.
Unfortunately the methods appearing in the literature were too quantitatively weak to apply in our context,\footnote{We are working with sparse random graphs, and the OGP does not hold when there is an edge between $s,t$, which happens with probability $q = \Omega(\log n / n)$. 
The works \cite{GJW24,Wein22} need OGP to hold with probability $1-1/n^{\Omega(D)}$ to derive a nontrivial conclusion for polynomials of degree $D$.}
but by appealing to symmetry and an invariance principle \cite{caravenna2023critical}, we are able to show that \pref{thm:ogp-intro} implies the following lower bound for low-degree polynomials.
\begin{proposition}[Informal version of \pref{thm:nodegreeDpoly}]\label{prop:ld-intro}
For any fixed $D\in \N$, there exists $\delta = 2^{-O(D)}$ such that any degree-$D$ polynomial approximation of the shortest $(s,t)$-path in $\bbG(n,q)$ fails with probability $\ge \delta$.
\end{proposition}
We believe that our invariance principle-based techniques may be useful for improving low-degree polynomial lower bounds in other sparse random models (though we are cognizant of the fact that in light of our other results, this may be a moot point).

The lower bound from \pref{prop:ld-intro} is typical of low-degree lower bounds derived from OGP, in that one only rules out degree-$D$ algorithms with exponentially small failure probability in $D$.
But in most contexts, one would be perfectly content with a polynomial-time computable degree-$O(\log n)$ algorithm which succeeds with probability $\Omega(1)$. 
The shortest $(s,t)$-path problem witnesses this to be more than just a weird technicality:

\begin{lemma}[Informal version of \pref{lem:exact-ld-2}]
    There is a degree-$O(\frac{\log n}{\log\log n})$ efficiently computable polynomial which exactly computes the indicator that edge $(i,j)$ participates in a near-shortest $(s,t)$-path in $\bbG(n,C\frac{\log n}{n})$ with success probability $1-o(1)$.
\end{lemma}

Our result highlights a potential brittleness of OGP lower bounds. 
The OGP implies unconditional lower bounds, but the subtle issue is that it only rules out algorithms with ultra-high success probability. 
Previously it was plausible that this was simply an artifact of the proof technique.
Our results demonstrate that sometimes there are indeed algorithms which, despite being smooth, succeed with decent probability, even where lower bounds rule out ultra-high success probability.

Two short months after our manuscript was posted on arXiv, the work \cite{HS25} has partially addressed this concern, showing that in some cases, if the OGP holds with probability $1-f(n)$, then using a new argument one can rule out algorithms of success probability $\Omega(1)$ and degree $o(\log\frac{1}{f(n)})$.

\paragraph{Dense models.}
Sparse average-case models are known to sometimes exhibit anomalous behavior.
In \pref{lem:ogp-fpp-single} we establish that the same overlap-gap phenomenon exists for shortest $(s,t)$-path in the dense average-case model wherein the complete graph $K_n$ has i.i.d. Exponential edge weights, also known as \emph{first passage percolation}.
We show that this implies lower bounds against stable algorithms (\pref{lem:fpp-nostable}).

\paragraph{Sampling.} 
OGP lower bounds have played a role in inspiring the related notion of ``disorder chaos,'' which roughly asks that if $\bG,\bG'$ are correlated samples from $\bbG(n,q)$, then the uniform distributions $\pi,\pi'$ over their $(1+\eps)$-approximate shortest paths are far apart in Wasserstein distance.
This property is known to imply lower bounds for sampling from $\pi$ with smooth algorithms, such as Langevin Dynamics or Approximate Message Passing.
In \pref{lem:disorder} we show that shortest $(s,t)$-path in sparse random graphs does indeed have disorder chaos, and therefore cannot be sampled by smooth algorithms.
On the other hand, as we observe in \pref{lem:sampling}, the set of all approximate $(s,t)$-shortest paths can be enumerated in polynomial time, and therefore one can also sample in polynomial time.
To our knowledge this is the first known example of a problem which exhibits disorder chaos, but for which sampling is easy.

\subsection{Discussion}
The OGP gives unconditional lower bounds against smooth and stable algorithms, with hardness results that qualitatively match known algorithmic thresholds for a number of average case problems.
But in the wake of our result, a number of questions need to be reconsidered.

\begin{enumerate}
    \item \emph{Are stable algorithms even appropriate for non-planted average case optimization problems?} 
    
    The OGP is only appropriate for optimization problems where there are multiple near-optimal solutions, and furthermore the near-optimal solutions themselves are brittle to the addition of noise.
    Given this, it seems clear that smooth and stable algorithms are actually ill-suited to these tasks from the start.
This is in sharp contrast to planted models, where we expect the planted optimal solutions to be noise stable.
    Perhaps we should be trying to rule out non-stable algorithms, such as linear programs and semidefinite programs, instead.
    
    At least in planted average-case models, logarithmic-degree polynomial bounds (so far) appear to give a good indication of the computational complexity of optimization (see e.g. \cite{hopkins2018statistical,KWB19}).
    Under the best of circumstances, polynomials of degree-$D$ are only $\exp(\Theta(D))$-smooth, so it is plausible that degree-$O(\log n)$ polynomials are a reasonable model to focus on.
In our original arxiv preprint, we asked whether one can rule out degree-$O(\log n)$ algorithms which succeed with constant probability.
Two months later, the work of \cite{HS25} gives an affirmative answer for many problems which exhibit the OGP, so long as the OGP holds with probability $1-o(1/\log n)$ and the ``ensemble OGP'' argument has Markovian structure.
Their results apply to e.g. maximum independent set in $\bbG(n,\frac{d}{n})$, and to max clique in $\bbG(n,\frac{1}{2})$.

    \item \emph{Is the class of degree-$O(\log n)$ polynomials expressive in the context of non-planted optimization? Concretely, can an algorithm for finding a $.9\log n$-sized clique in $\bbG(n,\frac{1}{2})$ be implemented as a degree-$O(\log n)$ polynomial?}

It is still not clear how powerful degree-$O(\log n)$ polynomials are in the context of non-planted average case optimization problems.
The best algorithms known for spin glasses, $k$-SAT, and max-independent set in $\bbG(n,\frac{d}{n})$ can be implemented by low-degree polynomials \cite{BLM15, Mon21,IS24,BH22,Wein22}.
However, we do not know how to implement the state-of-the art algorithm for max clique in $\bbG(n,\frac{1}{2})$ in the $O(\log n)$-degree model.
To check if OGP lower bounds are, in general, meaningful for non-planted optimization problems, we should make sure that the algorithms ruled out by OGP in the ``hard regime'' are actually useful in the easy regime.
    
    \item \emph{Are there additional conditions under which we expect that OGP is a good heuristic for hardness?}

    As was mentioned above, problems with algebraic structure, like $k$-XOR, are widely considered to be ``black sheep'' in average-case complexity, and so heuristics like OGP are not considered to be indicative of the underlying computational complexity.
    
Though $(s,t)$-path in unweighted graph can be solved by Gaussian Elimination, \emph{shortest} path cannot \cite{BMT78}.
    Still, the average-case shortest $(s,t)$-path problem has some structural features that set is apart from e.g. Ising models or max independent set.
    Salient differences include: (i) the size of the set of approximate solutions is $\poly(n)$ rather than exponential or subexponential; (ii) the probability of the overlap-gap property is only $1-O(1/\polylog n)$ rather than $1-O(1/\poly(n))$; (iii) an $(s,t)$ path is subject to the ``hard constraint'' that it actually be an unbroken path, whereas in e.g. Ising models there are no such constraints.
    Perhaps one can attribute the failure of OGP to accurately predict hardness to one of these structural features, or to some other property of $(s,t)$-paths?
The new work of \cite{HS25} suggests that (ii) might play a role, at least in some models.

    It seems that in order to understand if the OGP is a good heuristic prediction of hardness (beyond the concrete lower bounds one is able to prove from it), one must understand which features of shortest path distinguish it from other problems. 
    A possible way to go about this is to revisit the ``easy'' problems from Algorithms 101 one at a time, and check if each of them has the OGP.
\end{enumerate}

\subsection{Technical overview}
The proofs of most of our results are quite short.
For the convenience of the reader, we give a high-level sense of the arguments here.

\paragraph{Overlap-gap property.} We establish the overlap-gap property for $(s,t)$-paths in $\bbG(n,q)$ with $q = \frac{C\log n}{n}$ using a simple first moment argument.
We want to show that it is unlikely that two near-shortest paths from $s$ to $t$ overlap in a constant fraction of their edges.

\begin{wrapfigure}{r}{0.2\textwidth} 
    \centering
\includegraphics[width=\linewidth]{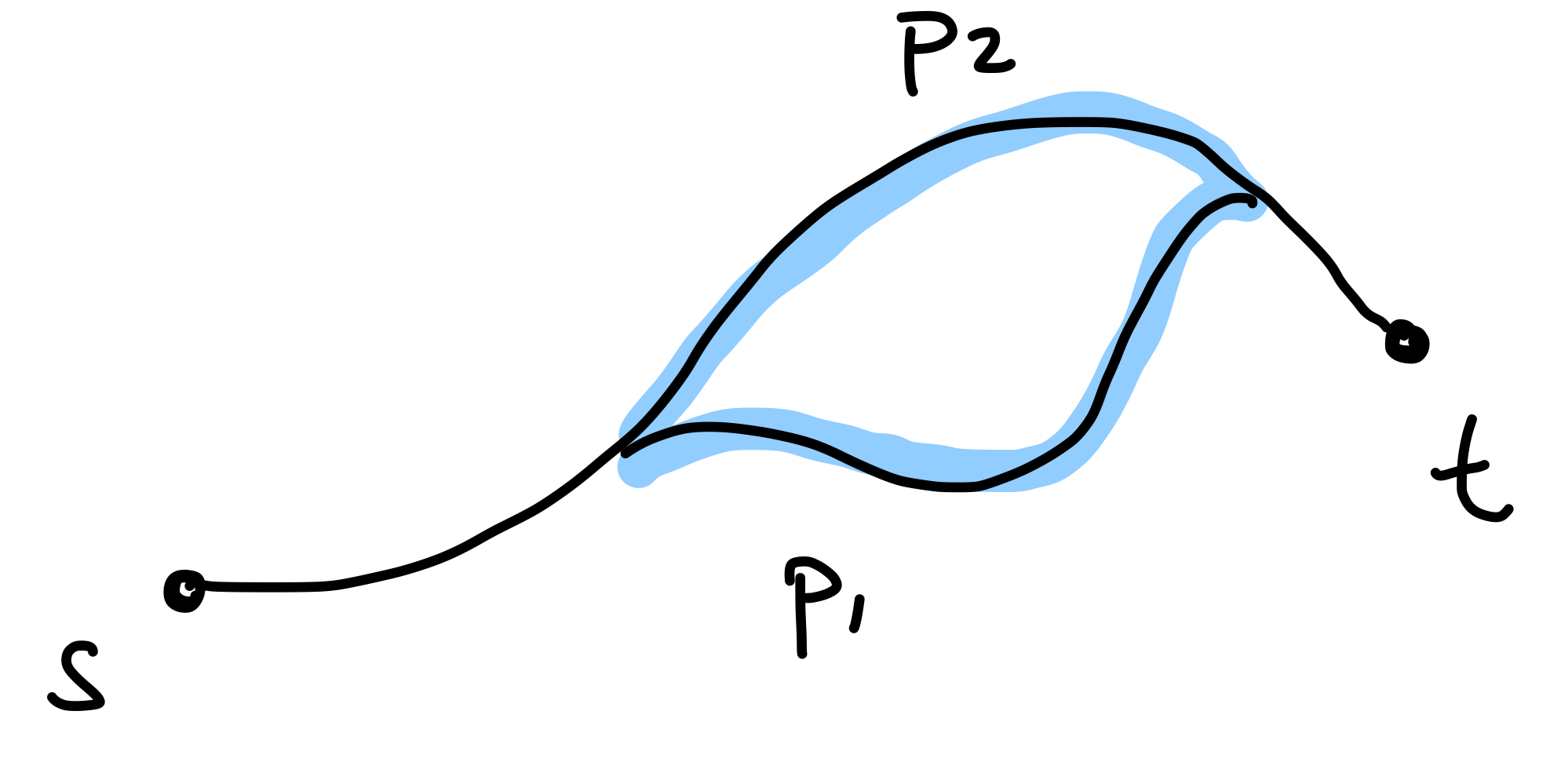}
\caption{Paths overlapping on a constant fraction of edges create a short cycle.}
    \label{fig:overlap}
\end{wrapfigure}

The reason we expect this to be true is that graphs from $\bbG(n,q)$ have only a few short cycles, and if there were two $s,t$ paths with near-minimum length which overlap in only, say, $50\%$ of their edges, then this means there must be a short cycle in the neighborhood of $s$ or $t$, which is not very 
likely.
Formalizing this amounts to some relatively simple counting arguments.
The existence of near-disjoint near-optimal paths comes from a second moment method lower bound on the number of near-optimal $s,t$ paths involving vertices only in $[n/2]$, and those involving vertices only in $\{n/2+1,\ldots,n\}$.
The arguments in First Passage Percolation are similar.

\paragraph{Stability of low-degree polynomials.}
In previous works, it was shown that low-degree polynomials in sparse random graphs are stable with probability $\Omega(1/\poly(n))$ \cite{Wein22,GJW24}.
This is too weak for us, as our OGP can never hold with probability larger than $1-O(\log n /n)$; if there happens to be an edge between $s$ and $t$, all bets are off.

To show that degree-$O(1)$ polynomials in the entries of the adjacency matrix of $\bG \sim \bbG(n,q)$ are stable, we prove an invariance principle for arbitrary \emph{symmetric} polynomials in correlated sparse Bernoulli random variables.
Prior work \cite{Wein22,GJW24} did not exploit symmetry; however the symmetry of the model $\bbG(n,q)$ under vertex relabeling implies that any asymmetric polynomial can be symmetrized without hurting its performance.
Furthermore, symmetric polynomials are more stable.
To show stability, we are able to appeal to the invariance principle/interpolation argument of \cite{caravenna2023critical}, which applies to functions of sufficiently low influence, or rather, functions whose gradient has sufficiently low $\ell_4$ norm.
Bounding the gradient of an arbitrary symmetric function of $\bG$ requires some nontrivial  combinatorial arguments.

\paragraph{Low-degree algorithms and sampling.}
The length of the shortest path in $\bG \sim \bbG(n,q)$ is $L = o(\log n)$ with high probability.
The indicator for the existence of a specific path $v_1,\ldots,v_L$ of length $L$ in $\bG$ can be written as a degree $L-1$ polynomial: $\prod_{i=1}^{L-1} \Ind[\bG_{v_i,v_{i+1}} = 1]$.
A good estimator for the indicator that edge $(u,v)$ participates in the shortest path from $s$ to $t$ is the sum over the indicators of all paths of length $L$ from $s$ to $t$ which include $(u,v)$.
Formalizing this requires some easy concentration arguments.

Sampling from the set of near-shortest paths can be done in polynomial time because the number of paths of length $(1+\eps)L$ starting at vertex $s$ is polynomial in $n$ with high probability.
One simply enumerates all of these paths, deletes the ones that do not end in $t$, and then samples one uniformly at random.

\subsection{Notation} 
    Throughout, we will use boldface font, such as $\bX$, to denote that $\bX$ is a random variable.
    We will use $n^{\ul{k}}$ to denote the falling factorial, $n^{\ul{k}} = n(n-1)\cdots(n-k+1)$.
    We use $\bbG(n,q)$ to denote the Erd\H{o}s-R\'{e}nyi distribution; almost everywhere, we will be in the setting where $q = \frac{C\log n}{n}$ for $C > 1$ a constant independent of $n$.
\section{Shortest path in a random graph}

In this section, let $\bG \sim \bbG(n, q)$, with $q =\frac{C\log n}{n}$ and $C > 1$.
We will study the shortest $(s,t)$-path problem in $\bG$.
By symmetry of $\bbG(n,q)$, we can assume $s = 1$ and $t = 2$ without loss of generality.

With high probability, the shortest path between vertices $1$ and $2$ has length $\OPT = (1+o(1))\frac{\log n}{\log nq}$ (we'll show as much, using the second moment method).
Furthermore, we will show in \pref{sec:ogp-gnp} that if we let $\calP_\eps(\bG)$ be the set of all paths of length $(1+\eps)\OPT$, then with high probability $\calP_\eps(\bG)$ has the overlap-gap property: each pair of paths $p_1,p_2 \in \calP_\eps(\bG)$ overlaps on either an $O(\eps)$ fraction of edges or on all of the edges; that is, there exists a constant $C$ such that for all $\eps$ sufficiently small,
\[
\frac{|p_1\cap p_2|}{\sqrt{|p_1|\cdot|p_2|}} \in [0,C\eps) \cup \{1\},
\]
and further there exist paths $p_1,p_2$ which are almost disjoint.
By-now standard arguments then imply that any sufficiently ``smooth'' algorithm cannot reliably find shortest paths in $\bbG(n,q)$.
This demonstrates that efficient optimization methods can be successful for mean-field optimization problems, even in the presence of an overlap-gap structure.

\medskip
Furthermore, we will show that the same is true for low-degree polynomials.
For average-case optimization problems, especially planted problems, the best degree-$O(\log n)$ polynomial estimators (such as spectral algorithms), anecdotally, achieve the same computational thresholds as any polynomial-time algorithm. 
For this reason, lower bounds against degree $\omega(\log n)$ polynomials have become a heuristic for predicting information-computation gaps.

The OGP also gives lower bounds against degree-$O(1)$ polynomial estimators, because they behave smoothly on average-case inputs.
We will show this is the case in \pref{sec:ld-gnp}; unfortunately, our OGP only holds with probability $1-1/\polylog n$, so we will not be able to apply black-box arguments and we will need to do some work in order to apply an invariance principle (we do this in \pref{sec:concentration}).
Though better than nothing, this is weaker evidence than a lower bound against a polynomial of degree-$\omega(\log n)$.
In \pref{sec:ld-gnp} we will also show that the shortest path problem constitutes a cautionary example, by demonstrating that polynomials of degree $\Theta(\frac{\log n}{\log\log n})$ can indeed approximate the shortest path in $\bbG(n,q)$, despite the presence of an overlap gap.

Finally, we will argue in \pref{sec:disorder} that the uniform distribution over $\calP_\eps(\bG)$ exhibits disorder chaos, but that sampling from this distribution is easy because one can enumerate $\calP_\eps(\bG)$ in polynomial time.

\subsection{Characterizing approximate shortest paths in Erd\"{o}s-R\'{e}nyi graphs}

We begin by characterizing the length, $\OPT$, of the shortest path in $\bbG(n,q)$, and the number of $(1+\eps)\OPT$-length paths in $\bbG(n,q)$.

\begin{lemma}[Shortest path in $\bbG(n,q)$]\label{lem:shortest-gnp}
If $\bG \sim \bbG(n,\frac{C\log n}{n})$, then with probability $\ge 1-O(\frac{1}{\log^d n})$ the length of the shortest path between $1,2$ in $\bG$ is $\frac{\log n}{\log\log n + \log C} \pm d$, and with probability $1-O(\frac{\log\log n}{\log n})$, $|\calP_\eps(\bG)| = (1\pm o(1)) n^\eps$.
\end{lemma}
\begin{proof}
    Let $N_m(\bG)$ be the number of paths of length $m$ between $1,2$ in $\bG$.
    In expectation,
    \[
        \E[N_m(\bG)] 
        = n^{\ul{m-1}} q^m 
        \le \frac{1}{n}(nq)^m.
    \]
    Hence by Markov's inequality, if $m < \frac{\log n - \log \frac{1}{\delta}}{\log nq}$, $N_m(\bG) = 0$ with probability $\ge 1-\delta$.
    This gives the probable upper bound on the length of the shortest path.
    Henceforth, define $\ell^* = \frac{\log n}{\log nq}$.

    Now, we argue that for $m = (1+\eps)\ell^*$, $N_m$ has expectation $\sim n^\eps$, and concentrates around its expectation.
    Since $n^{\ul{m}} \ge n^{m} (1-\frac{m}{n})^m \ge n^{m}(1-\frac{{m}^2}{n})$, we have that for $m \le \log n$,
    \[
    \E[N_{m}(\bG)] 
    \ge \frac{(nq)^{m}}{n}\left(1-\tfrac{\log^2 n}{n}\right)
    \ge n^{m/\ell^* -1}\left(1-\tfrac{\log^2 n}{n}\right)
    \ge n^\eps \left(1-\tfrac{\log^2 n}{n}\right).
    \]
    
    We now bound the variance of $N_m$.
    Let $\calP_{(m)}$ be the set of all paths of length $m$ between $1,2$ in $K_n$.
    Then
    \begin{align}
    \E[N_{m}^2] 
&= \sum_{p_1,p_2 \in \calP_{(m)}} \Pr[p_1,p_2 \in \bG]
= \sum_{p_1,p_2 \in \calP_{(m)}} q^{2m - |p_1 \cap p_2|},\label{eq:var}
    \end{align}
    where the notation $|p_1 \cap p_2|$ refers to the number of common edges.
    We will parameterize the above sum according to the number of edges in $p_1\cap p_2$, and according to the number of connected components in $p_1 \cap p_2$.
    We give a more general characterization than we need here, as it will be useful in later proofs.
    \begin{claim}\torestate{\label{claim:path-pairs}
        Let $M_{k,d_1,d_2}$ be the number of pairs of paths $p_1,p_2$ with $|p_1| = d_1 + k$, $|p_2| = d_2 + k$, and $|p_1\cap p_2| = k$ in $K_n$, with $d_1,d_2 > 0$.
        Then if $k,d_1,d_2 \ll n^{1/3}$, for all $n$ sufficiently large, 
        \[
        M_{k,d_1,d_2} \le \left(\frac{k+1}{n^2} + \left(\frac{100 k^3 (d_1+k)(d_2+k)}{n}\right)^3\right)n^{d_1+d_2+k}.
        \]}
    \end{claim}
    \begin{proof}
        Suppose that the intersection graph $p_1\cap p_2$ has $c \ge 2$ components (at least two because $d_1,d_2 > 0$, and the endpoints $1,2$ are present in both paths), and $k$ edges.

        First, if $c = 2$, then one can only choose how many edges to include in the component containing the source vertex, $1$. 
        There are at most $k+1$ choices for this.
        Then, there are $d_1 + d_2 +2k - (k+c) = d_1+d_2+k-c$ choices for the vertex labels,
        giving a total of at most $(k+1) n^{d_1+d_2+k-2}$ such graphs.

        Now, suppose that $c > 2$, and suppose that $\delta$ of the components contain no edges.
        By a ``stars and bars'' argument, the number of choices for the number of edges in each component is $\binom{k-1}{c - \delta -1}$.
        Since $p_1,p_2$ must intersect on sub-paths, the edges participating in each common component can be chosen by (i) choosing if the component appears in lexicographic or reverse-lexicographic order in $p_2$ (with respect to $p_1$), (ii) choosing a first edge in $p_1$ and either a first or last edge in $p_2$, depending on the lexicographic choice.
        The number of such choices is $2^{c-\delta} \cdot \binom{d_1+k}{c-\delta} (d_2+k)^{\ul{c-\delta}}$.
        Now, label the vertices: since there are $c$ components and $k$ edges shared in $p_1\cap p_2$, and the intersection forms a tree, there are $k+c$ vertices in $p_1 \cap p_2$.
        So the total number of vertices to label is $d_1 + d_2+ 2k - (k+c) = d_1 + d+2 + k - c$, to avoid double-counting overlapping vertices (and accounting for the fact that the endpoints $1,2$ are already labeled).
        
        This counts all valid overlapping paths, and may overcount since sometimes the components end up overlapping.

        The total number is thus at most
        \[
        \binom{k-1}{c-\delta-1}\cdot 2^{c-\delta}(c-\delta)!\binom{d_1+k}{c-\delta}\binom{d_2+k}{c-\delta} \cdot n^{d_1 + d_2 + k - c}
        \le \left(\frac{2e^3k(d_1+k)(d_2+k)}{n}\right)^c n^{d_1+d_2 + k},
        \]
        where we have used Stirling's approximation.

        Summing over all $\le k$ choices of $c$ and all $\le k$ choices of $\delta$, for $n$ sufficiently large,
        \[
        M_{k,d_1,d_2} \le \left(\frac{k+1}{n^2} + \left(\frac{100 k^3 (d_1+k)(d_2+k)}{n}\right)^3\right)n^{d_1+d_2+k}. \qedhere
        \]
    \end{proof}
    Returning to \pref{eq:var}, we have
    \begin{align*}
    \E[N_m^2]
    &\le n^{2m - 2}q^{2m} + n^{m-1}q^m + \sum_{k=1}^{m-1} M_{k,m-k,m-k}\cdot q^{2m - k} \\
    &\le \left(1+O(\tfrac{1}{n^\eps})\right) n^{2m - 2}q^{2m} + \sum_{k=1}^{m-1} \left(\frac{k+1}{n^2} + \frac{\log^{15} n}{n^3}\right)(nq)^{2m - k}\\
    &\le \left(1+O(\tfrac{1}{n^{\eps}})\right)n^{2m - 2}q^{2m} + n^{2m-2}q^{2m}\sum_{k=1}^{m-1} \left(\frac{k+1}{(nq)^k} + \frac{\log^{15} n}{n(nq)^k}\right)\\
    &\le  \left(1+O(\tfrac{1}{n^\eps})\right)n^{2m - 2}q^{2m} + O\left(\tfrac{1}{\log n}\right) \cdot n^{2m-2}q^{2m}
    = \left(1+O(\tfrac{1}{\log n})\right) \E[N_m]^2.
    \end{align*}
    Hence applying Chebyshev's inequality we have that $N_{(1+\eps)\ell^*} \ge \frac{1}{2} n^\eps$ with probability $\ge 1 - O(\frac{1}{\log n})$.

    From the concentration of the $N_m$ quantity and $m = (1+O(\frac{1}{\log n}))\ell^*$ and $(1+\eps)\ell^*$, we conclude that $\OPT \le \ell^* + 1 = (1+O(\frac{1}{\log\log n}))\ell^*$ with probability $\ge 1-O(\frac{\log\log n}{\log n})$, and that $|\calP_{\eps}(\bG)|= (1\pm o(1))n^\eps$ with probability $\ge 1-O(\frac{\log\log n}{\log n})$.
\end{proof}

\subsection{The overlap-gap property in random graphs}\label{sec:ogp-gnp}
Now we will use the second moment method to show that shortest paths in $\bbG(n,q)$ have the overlap-gap property. 
We will also derive a lower bound against any sufficiently stable algorithm.

\begin{theorem}\label{thm:ogp-gnp}
Let $\bG,\bG'\sim \bbG(n,q)$ for $q = \frac{B\log n}{n}$ and $B > 1$, where $\bG'$ is sampled from $\bG$ by re-sampling each edge or non-edge with probability $1-\rho$. 
Then there exists a constant $C$ such that for all $\eps > 0$ sufficiently small, with probability $1-O(\frac{\log\log n}{\log n})$ all pairs $(p,p') \in \calP_\eps(\bG)\times \calP_\eps(\bG')$ have 
\[
\frac{|p\cap p'|}{\sqrt{|p|\cdot|p'|}} \in \begin{cases}
    [0,C\eps) & \rho < (\frac{1}{\log n})^{2\eps}\\
    [0, C \eps) \cup \{1\} & \text{otherwise}.
\end{cases}
\]
Furthermore, with probability $\ge 1-O(\frac{\log\log n}{\log n})$ there exist $p,p'$ of overlap $< C\eps$.
\end{theorem}
\begin{proof}
Condition first on the outcome of \pref{lem:shortest-gnp}. 
Now, let $N_{m,m',k}(\bG,\bG')$ be the number of pairs $(p,p')$ where $p$ is a $1$-$2$ path of length $m$ in $\bG$, $p'$ is a $1$-$2$ path of length $m'$ in $\bG'$, $p \cap p'$ contains $k$ edges, and suppose first $k < m,m'$ (ruling out overlap 1).
Applying \pref{claim:path-pairs}, we have
\begin{align*}
\E[N_{m,m',k}(\bG,\bG')] 
&= M_{k,m-k,m'-k} \cdot q^{m + m'-k} \cdot \left(\rho + (1-\rho)q\right)^{k}\\
&\le \left(\frac{k+1}{n^2} + \frac{\log^{15} n}{n^3}\right)(nq)^{m + m'-k} \left(\rho + (1-\rho)q\right)^{k}\\
&\le 2\frac{k+1}{n^2}(nq)^{m + m'-k} \left(\rho + (1-\rho)q\right)^k\\
&\le 2 (k+1) n^{2\eps} \left(\frac{\rho + (1-\rho)q}{nq}\right)^k,
\end{align*}
where we have used that $m,m' \le (1+\eps)\frac{\log n}{\log nq}$ by \pref{lem:shortest-gnp}.
This is $O(1/\log^5 n)$ when $k > \frac{2\eps \log n + 5\log\log n}{\log nq}$.
Taking a union bound over all $O(\log^3 n)$ values of $\frac{\log n - 100}{\log nq} \le m,m' \le (1+\eps)\frac{\log n}{\log nq}$ and $m > k > 2\eps \frac{\log n}{\log nq}$ gives the first part of the result.

Now, suppose $k = m =m'$ (notice that if $k =m$, it must be the case that $m = m'$, since the endpoints of the paths are equal).
In this case, 
\[
\E[N_{m,m,m}] \le n^{m-1}q^m\left(\rho + (1-\rho)q\right)^m = \frac{(nq(\rho + (1-\rho)q))^m}{n}
\]
which is $O(\frac{1}{\log n})$ when $\rho \le \frac{1}{(\log n)^{2\eps}}$.

The final remark, regarding the existence of low-overlap pairs, follows from the lower bound on $|\calP_\eps(\bG)|$ in \pref{lem:shortest-gnp}, along with the following observation: if one partitions $[n]\setminus\{1,2\}$ into two disjoint and equally-sized sets of vertices, $A$ and $B$, then \pref{lem:shortest-gnp} assures us that with high probability there will be $(n/2)^{\Omega(\eps)}$ near-optimal paths with all vertices in $A$, and similarly in $B$. 
These paths will be disjoint.
\end{proof}

The overlap-gap property allows us to rule out stable algorithms for shortest path:

\begin{corollary}\label{cor:gnp-no-stable}
    Let $\rho \in [0,1)$ be bounded away from $1$, $\eps > 0$ be sufficiently small, and $n$ be sufficiently large.
    Then there can be no algorithm for $\eps$-approximate shortest path which simultaneously (i) has failure probability $\le \frac{1-\rho}{6}$, and (ii) is stable, in the sense that if $\bG,\bG' \sim \bbG(n,\frac{C\log n}{n})$ and the edges of $\bG'$ are at least $\rho$-correlated with the edges of $\bG$, then conditioned on $\calA$ succeeding on its inputs, $\frac{|\calA(\bG) \cap \calA(\bG')|}{\sqrt{|\calA(\bG)|\cdot|\calA(\bG')|}} > C\eps$.
\end{corollary}
\begin{proof}
Choose $T = \lceil \frac{1}{1-\rho} \rceil = O(1)$, and note $\rho \ge 1-\frac{1}{T}$.
We sample a sequence of graphs $\bG_0,\bG_1,\ldots,\bG_T$, with each $\bG_t \sim \bbG(n,q)$, $\bG_0,\bG_T$ are independent, and each pair $\bG_t,\bG_{t+1}$ is marginally $(1-\frac{1}{T})$-correlated (that is, $\bG_t,\bG_{t+1}$ can be coupled with a pair $\bG,\bG'$ with $\bG \sim \bbG(n,q)$ and $\bG'$ obtained by resampling each edge with probability $\frac{1}{T}$).
First, sample $\bG_0,\bG_T$ independently.
Additionally, for each $(i,j) \in \binom{[n]}{2}$, sample an independent random variable $\bU_{ij} \sim \Unif([0,1])$.
Now, for each $t$, let $\bG_t(i,j) = \bG_0(i,j) \cdot \Ind[ \bU_{ij} > \frac{t}{T}] + \bG_T(i,j) \cdot \Ind[\bU_{ij} \le \frac{t}{T}]$.
Clearly, $\bG_t \sim \bbG(n,q)$.
Also, edge $(i,j)$ is resampled going from $\bG_t$ to $\bG_{t+1}$ if and only if $\bU_{ij} \in (\frac{t}{T},\frac{t+1}{T}]$, which is true for each edge independently with probability $\frac{1}{T}$.
Thus, $\bG_t,\bG_{t+1}$ are marginally a $1-\frac{1}{T}$-correlated pair.  

Let $\bp_t = \calA(\bG_t)$, and suppose $\calA$ fails with probability $\delta$.
From \pref{thm:ogp-gnp}, we have that with probability at least $1-O(\frac{T\log\log n}{\log n}) \ge \frac{1}{2}$, each $\bp_t$ must have overlap with $\bp_0$ which is either equal to $1$, or at most $C\eps$, and further the overlap of $\bp_0$ and $\bp_T$  must be at most $C\eps$ (as they are independent).
Hence, with high probability there must exist some $t \in [T]$ with 
$\frac{|\bp_0 \cap \bp_{t}|}{\sqrt{|\bp_0||\bp_{t}|}} = 1$ but
$\frac{|\bp_0 \cap \bp_{t+1}|}{\sqrt{|\bp_0||\bp_{t+1}|}}\le C\eps$, 
implying that $\frac{|\bp_t \cap \bp_{t+1}|}{\sqrt{|\bp_t||\bp_{t+1}|}} \le C \eps $.
This is a contradiction, unless $\calA$ did not succeed on some input, so the success probability on all inputs cannot exceed $\frac{1}{2}$.
By a union bound $\calA$ must have been successful on all $T+1$ inputs with probability at least $1-\delta(T+1)$, so it must be the case that $\delta > \frac{1}{2(T+1)} \ge \frac{1-\rho}{6}$.
\end{proof}

\subsection{Low-degree polynomial estimators}\label{sec:ld-gnp}
As discussed above, the shortest $s,t$-path problem can be solved, even exactly, in polynomial time.
The overlap-gap property means that it cannot be solved by any sufficiently smooth algorithm.
In particular, in \cite{GJW24,Wein22} it is shown that polynomials of degree $O(1)$ are smooth with probability $\Omega(1/\poly(n))$ when the input is an adjacency matrix of a graph sampled from $\bbG(n,q)$.\footnote{They also observe that this is the case for e.g. polynomials in Normal-valued variables, where the fact is just a consequence of hypercontractivity. 
In sparse random graphs, matters are a bit more complex.
}
Here, we will apply an invariance principle to conclude that symmetric degree-$O(1)$ polynomials are smooth with probability $\Omega(1)$, which will allow us to prove a lower bound against degree-$O(1)$ polynomials.

But while degree $O(1)$ polynomials are smooth, degree $\omega(1)$ polynomials are not necessarily smooth.
In \pref{sec:lb-ld} we will show that the OGP indeed implies that degree-$O(1)$ polynomials cannot estimate $s,t$-paths, and in \pref{sec:ub-ld} we'll show that polynomials of degree $\Omega(\log n/\log\log n)$ can estimate the shortest $s,t$-path problem very well.

We identify each path $p$ from $1$ to $2$ by its indicator vector in $\{0,1\}^{\binom{[n]}{2}}$.
\begin{definition}\label{def:polyest}
We say that $f:\{0,1\}^{\binom{[n]}{2}}\to\R^{\binom{n}{2}}$ is an \emph{$(\alpha,\beta)$-approximation of the $(1+\eps)$-approximate shortest path between $1,2$} in $\bG \sim \bbG(n,q)$ if $\E[\|f\|^2] = \OPT$,\footnote{This is just a normalization convention; the parameters are easily adjusted for rescaling by $(1+\eps)$.} and
\[
\Pr\left[ \exists p \in \calP_\eps(\bG) \,\, \text{s.t.}\,\, \frac{\langle f,p\rangle}{\sqrt{\E[\|f\|^2]}\cdot\|p\|} \ge \alpha, \quad \text{and}\quad \alpha \le \frac{\|f\|^2}{\E[\|f\|^2]} \le \frac{1}{\alpha}\right] \ge \beta.
\]
\end{definition}

\subsubsection{Lower bound for constant degree polynomials}\label{sec:lb-ld}
Degree-$O(1)$ polynomial estimators are smooth algorithms.
This has already been established in the literature, but extant results such as \cite{GJW24,Wein22} establish that stability holds with probability at least $1/\poly(n)$; here, because our overlap-gap property holds only with probability $1-\frac{1}{\polylog n}$, we require a new statement which takes advantage of the fact that an estimator $f$ is, without loss of generality, symmetric under vertex-relabelings.

\begin{lemma}\label{lem:sym}
    Suppose $f:\{0,1\}^{\binom{[n]}{2}}\to \R^{\binom{[n]}{2}}$ is an approximation to the shortest $1,2$-path.
    For each $\pi \in S_n$, define $f_\pi(G) = \pi^{-1} f(\pi(G))$, and let $f^{\sym}(G) = \E_{\pi \sim \Unif(S_{[n]\setminus \{1,2\}})} \pi^{-1} f (\pi (G))$.
    Then for any coupling of graphs $\bG \sim \bbG(n,q)$ and $s-t$ paths $\bp$ in $\bG$,
    \[
    \frac{\E\langle f(\bG), \bp \rangle}{\sqrt{\E\|f(\bG)\|^2}} \le \frac{\E\langle f^\sym(\bG), \bp \rangle}{\sqrt{\E\|f^\sym(\bG)\|^2}}.
    \]
\end{lemma}
\begin{proof}
The numerators are the same, because $\bG$ and its $1-2$ paths have a distribution which is symmetric under permutations of $[n]$ which fix $1,2$, so 
\[
\E_{\bG,\bp}\langle f(\bG),\bp\rangle = \E_{\bG,\bp}\E_\pi \langle f(\pi(\bG)),\pi(\bp)\rangle = \E_{\bG,\bp} \langle f^\sym(\bG),\bp\rangle.
\]
The denominator on the right-hand side is only smaller by Jensen's inequality.
\end{proof}

If $f$ is a $(1-\eta,1-\delta)$-approximation of the $(1+\eps)$ shortest $s,t$-path for $\eta,\delta,\eps$ sufficiently small, then by an averaging argument $\langle f(\bG),\bp \rangle$ will come close to saturating the Cauchy-Schwarz inequality with good probability. 
In this case, \pref{lem:sym} implies that $f^\sym(G) = \E_\pi \pi^{-1}(f(\pi(G)))$ is at least as close to saturating Cauchy-Schwarz in expectation, from which another averaging argument implies that $f^\sym$ is also a $(1-\eta',1-\delta')$-approximation to the $s,t$-shortest path, for $\eta',\delta'$ going to zero with $\eta,\delta$.
Since we will only be concerned with the small $\eta,\delta$ regime, from now on without loss of generality we consider only symmetric $f$, as ruling out $f^\sym$ suffices to rule out any $f$.
 
Symmetric functions are more stable than asymmetric functions, which allows us to give better quantitative guarantees than \cite{GJW24,Wein22}.
We will prove the following smoothness result in \pref{sec:concentration}.
\begin{theorem}\label{thm:smooth}
    For $\rho \geq 1-1/T$, $\gamma \geq 3D (6e)^D /T$, and $f$ a degree-$D$ polynomial which is fixed by the action of $S_{[n]\setminus\{1,2\}}$, for  $\rho$-correlated graphs $\bG,\bG' \sim \bbG(n,q)$, we have that
    \begin{align*}
        \Pr\left[ \|f(\bG)-f(\bG')\|_2^2\geq \gamma \E[\|f(\bG)\|^2] \right] \leq \exp\left(-\frac{D}{3 e} \left( \frac{\gamma T}{3D} \right)^{1 / D}\right)+ o_n(1). 
    \end{align*}
\end{theorem}

\paragraph{Lower bound from overlap-gap property}
\begin{theorem}\label{thm:nodegreeDpoly}
    For any integer $D\ge 0$ and $0<\eta<0.1$, there is a small enough constant $\delta>0$ such that there is no degree-$D$ polynomial which is a $(1-\eta,1-\delta)$-approximation of the shortest path between $1,2$ in $\bbG(n,q)$.
\end{theorem}
\begin{proof}
We give a proof by contradiction.
Suppose there exists a degree-$D$ polynomial which is a $(1-\eta,1-\delta)$-approximation of the shortest path between $1,2$ in $\bbG(n,q)$.
Choose $T = \lceil \frac{1}{1-\rho} \rceil = O(1)$, and note $\rho \ge 1-\frac{1}{T}$. 
We sample a sequence of graphs $\bG_0,\bG_1,\ldots,\bG_T$ as in the proof of \pref{cor:gnp-no-stable}, where each $\bG_t \sim \bbG(n,q)$, $\bG_0,\bG_T$ are independent, and each pair $\bG_t,\bG_{t+1}$ is marginally $(1-\frac{1}{T})$-correlated.
For shorthand, call $L = \frac{\log n}{\log\log n}$ to be the typical approximate length of the shortest path in $\bG$.
We define four events:
\newcommand{\Estable}{\calE_{\mathrm{stable}}}
\newcommand{\Egap}{\calE_{\mathrm{gap}}}
\newcommand{\Ealg}{\calE_{\mathrm{alg}}}
\newcommand{\Epaths}{\calE_{\mathrm{paths}}}
\begin{enumerate}
    \item Let $\Estable$ be the event that $\|f(\bG_t)-f(\bG_{t+1})\| \le s \sqrt{\E[\|f(\bG)\|^2]}$ for all $0\leq t\leq T-1$.
    \item Let $\Ealg$ be the event that $f$ is a $(1-\eta,1-\delta)$-approximation of a path in $\calP_\eps(\bG_t)$ for all $0\leq t\leq T$.
    \item Let $\Egap$ be the event that 
    for all $t\leq T-1$ and $(p_t,p_{t+1}) \in \calP_{\eps}(\bG_t) \times \calP_\eps(\bG_{t+1})$, $\frac{1}{L}\|p_t-p_{t+1}\|^2 \not\in (0,2-2C\eps)$, and that
    for all $(p_0,p_T) \in \calP_{\eps}(\bG_0) \times \calP_{\eps}(\bG_T)$, $\frac{1}{L}\|p_0-p_T\|^2 \ge 2- 2C\eps$.
    \item Let $\Epaths$ be the event that for all $t\leq T$, $\calP_\eps(\bG_t)\neq \emptyset$ and contains only paths of size in $[1,1+\eps]L$.
\end{enumerate}
We argue that these events cannot happen simultaneously for small enough constant $s$ and $\eps$.
Indeed, assume $\Epaths$ occurs, so that each set of paths is non-empty and the paths have the expected size; for shorthand, 
assume as well that $\Ealg$ occurs, so that with each output $f(\bG_t)$ we may associate a path $p_t$ with $\|f(\bG_t) - p_t\|^2 \le \|f(\bG_t)\|^2 + \|p_t\|^2 -2(1-\eta) \sqrt{\E \|f(\bG_t)\|^2 } \|p_t\| \le 4\eta L$. Assume as well that $\Egap$ occurs, so that $\frac{1}{L}\|p_0-p_T\|^2 \ge 2-2C\eps$, and for all $t$, $\frac{1}{L} \|p_t-p_{t+1}\|^2\not\in (0,2-2C\eps)$.
This implies there must be a first time $t^* \in [T]$ where $\frac{1}{L}\|p_0-p_{t^*}\|^2 \ge 2-2C\eps$.
But now if $\Estable$ occurs, by triangle inequality
\begin{align*}
\|p_0 - p_{t^*}\| 
&\le \|p_0 - p_{t^*-1}\| + \|p_{t^*-1} - p_{t^*}\|\\
&= \|p_{t^*-1} - p_{t^*}\|\\
&\le \|p_{t^*-1} - f(\bG_{t^*-1})\| + \|f(\bG_{t^*-1}) - f(\bG_{t^*})\| + \|f(\bG_{t^*}) - p_{t^*}\|\\
&\le 4\sqrt{\eta} \sqrt{L} + O(s \sqrt{L}) < \sqrt{(2-2C\varepsilon)L},
\end{align*}
for small constant $s$ and $\eps$, which is a contradiction.

Now, we lower-bound the probability of the events one-by-one.
By construction, each $\bG_t \sim \bbG(n,q)$ and each pair $\bG_t,\bG_{t+1}$ is marginally $(1-\frac{1}{T})$-correlated. 
By \pref{thm:smooth} and a union bound over pairs $\bG_t,\bG_{t+1}$, we have $\Pr[\ol{\Estable}] \le T \exp(-\frac{D}{3 e} ( \frac{s^2 T}{3D})^{1 / D})+ o_n(1)$.
The failure probability of the algorithm is at most $\delta$ by assumption, so by a union bound $\Pr[\ol{\Ealg}] \le (T+1)\delta$. 
Finally, since $\bG_0$ and $\bG_T$ are independent, from \pref{thm:ogp-gnp} and \pref{lem:shortest-gnp} we have that $\Pr[\ol{\Egap}] + \Pr[\ol{\Epaths}] \le (T+1) \cdot O\left(\frac{\log\log n}{\log n}\right)$. 
Thus, it must be the case that
\begin{align*}
0 
\ge \Pr[\Estable\cap \Ealg\cap \Egap \cap \Epaths]
&\ge 1 - \Pr[\ol{\Estable}] - \Pr[\ol{\Ealg}] - \Pr[\ol{\Egap}] -\Pr[\ol{\Epaths}] \\
&\ge 1 - T \exp\Big(-\frac{D}{3 e} \Big( \frac{s^2 T}{3D} \Big)^{1 / D}\Big)+ o_n(1) - (T+1)\cdot \delta - T \cdot O\left(\tfrac{\log\log n}{\log n}\right),
\end{align*}
the positivity of the right-hand-side for $\delta = \frac{1}{100(T+1)}$ and large enough $T$ yields a contradiction.
\end{proof}

\medskip

\subsubsection{A polynomial estimator of logarithmic degree}\label{sec:ub-ld}
Here we will argue that the shortest $s,t$-path may be computed with reasonable success probability by degree $o(\log n)$ polynomials.
This estimator does not, strictly speaking, fulfill \pref{def:polyest}, because of an issue of breaking symmetry.
However,it gives almost full information about whether an edge participates in a short $s,t$-path.\footnote{This raises another question about the usual OGP-based low-degree polynomial lower bounds: they may be brittle to the specific definition of a low-degree polynomial estimator.}
\begin{lemma}\label{lem:exact-ld-2}
There exists a degree-$O(\frac{\log n}{\log\log n})$ polynomial in the entries of the adjacency matrix of $\bG$ which achieves correlation $1-o(1)$ with the indicator that $(i,j)$ participates a path of length $m= \lceil \frac{\log n}{\log nq}\rceil + 1$ in $\bG$.
Furthermore, with high probability the polynomial is computable in polynomial time.
\end{lemma}
\begin{proof}
    Define $m = \lceil\frac{\log n}{\log nq}\rceil +1 = \frac{\log n}{\log nq} + 1+\xi$. 
    By \pref{lem:shortest-gnp}, the length of the shortest path in $\bG$ will be close to $m$, and with $1-o(1)$ probability there will be a $1$-$2$ path of length $m$. 
    Define $\bp_{ij}$ to be the indicator that $(i,j)$ participates in an $1$-$2$ path of length $m$. 
    Define the polynomial
    \[
        f_{ij}(\bG) = \sum_{\substack{p \in \calP_{(m)}\\ (i,j) \in p}} \prod_{(a,b) \in p} \Ind[(a,b) \in E(\bG)],
    \]
    where we recall that $\calP_{(m)}$ is the set of all $m$-edge $1$-$2$ paths in $K_n$.

    We now verify that $f_{ij}(\bG)$ correlates well with $\bp_{ij}$, in the sense that
    \begin{align}
    \label{eq:cor} \frac{\E[f_{ij}(\bG) \bp_{ij}]}{\sqrt{\E[f_{ij}(\bG)^2]\E[\bp_{ij}^2]}} = 1-o(1).
    \end{align}
    We must consider separately the case when $\{i,j\}\cap \{1,2\} = \emptyset$ from the case when $|\{i,j\}\cap\{1,2\}| = 1$; for brevity's sake we consider only the latter case, as the former case is similar (and actually simpler).
    
    To compute the correlation, we use that $f_{ij}(\bG)$ takes on non-negative integer values, and is zero only if $(i,j)$ participates in no $1$-$2$ paths of length $m$.
    Therefore
    \[
    \E[f_{ij}(\bG)\bp_{ij}] = \E[f_{ij}] = 2(m-2) n^{\ul{m-3}}\cdot q^m,
    \]
    as there are $2(m-2)n^{\ul{m-3}}$ paths from $1$ to $2$ that involve the edge $i,j$ (choose the location and orientation of $(i,j)$, then the remaining vertex labels). 
    By our choice of $m$, this quantity is $O(\frac{\log^{2+\xi} n}{n^2})$.

    To compute the $\E[\bp_{ij}^2] = \E[\bp_{ij}] = \Pr[(i,j) \text{ in an }m-\text{path}]$, we apply the Janson inequality (see e.g. Chapter 8 of \cite{alon2016probabilistic}):
    \begin{fact}[Corollary of Janson Inequalities]\label{fact:janson}
       If $\{H_\alpha\}_{\alpha}$ is a set of edge-induced subgraphs of $K_n$, the probability that none of the subgraphs appear in $\bbG(n,q)$ is at least $\prod_\alpha(1-q^{|E(H_\alpha)|})$.
    \end{fact}
    This implies that 
    \[
    \E[\bp_{ij}] 
    = 1 - \Pr[(i,j) \text{ not in }m-\text{path}] 
    \le 1 - (1-q^m)^{2(m-2)n^{\ul{m-3}}}
    \le (1+o(1))\cdot 2(m-2)n^{\ul{m-3}}q^m,
    \]
    where we have used that $q^{2m} \ll 1$.
    
    Finally, we compute the second moment of $f_{ij}$.
    For this, we will need to count the number of pairs of paths from $1,2$ that intersect on the edge $(i,j)$.
    \begin{claim}
        Let $W_{k,m}$ be the number of pairs of paths from $1$ to $2$ that overlap on a total of $k$ edges, one of which is the edge $(i,j)$.
        Then
        \[
        W_{k,m} \le \begin{cases}
           4(m-2)^2n^{2m-6} & k=1,\\
           k(k+1)n^{2m-k-4} + O\left(k^4 m^6\right) n^{2m-k-5} &k>1.
        \end{cases}
        \]
    \end{claim}
    \begin{proof}
        Let $c$ be the number of components in $p\cap p'$.
        As $k < m$, there must be at least two components.
        
        In the case $c = 2$, only $k > 1$ is possible, since we have assumed $\{i,j\}$ does not intersect with $\{1,2\}$. 
        So one chooses $k+1$ possibilities for the size of the component containing $1$, $k$ possibilities for the edge $(i,j)$, and $2m - k - 4$ vertex labels, for a total of $k(k+1)n^{2m - k - 4}$ choices.

        In the case $c = 3$, $k=1$, there are $4(m-2)^2n^{2m-6}$ pairs: choose the location of $i,j$ in each path, the orientation of $i,j$ in each path, and then all the remaining vertex labels.

        Finally, if $c \ge 3$ we argue just as in \pref{claim:path-pairs}, except that we choose one of the $k$ edges to be $(i,j)$, choose its $2$ orientations, and then lose $n(n-1)$ choices of vertex labels since $i,j$ are fixed.
        Hence summing over all $c \ge 3$, in this case there are at most $O(k^4m^6n^{2m-k-5})$ paths.
    \end{proof}
    Thus by \pref{claim:path-pairs},
    \begin{align*}
        &\E[f_{ij}(\bG)^2] 
        =\sum_{(i,j) \in p,p' \in \calP_{(m)}} q^{2m - |p\cap p'|}\\
        &\le 2(m-2)n^{m-3} q^{m} 
        + 4(m-2)^2 n^{2m-6} q^{2m-1}
        +\sum_{k=2}^{m-1} W_{k,m} q^{2m-k}\\
        &\le 2(m-2)n^{m-3} q^{m} 
        + 4(m-2)^2 n^{2m-6} q^{2m-1}
        + \sum_{k=2}^{m-1} \left(k(k+1)n^{2m-k-4} + O(k^4m^6)n^{2m-k-5}\right) q^{2m-k}\\
        &\le 2(m-2)n^{m-3} q^{m} 
        + 4(m-2)^2 n^{2m-6} q^{2m-1}
        + n^{m-3}q^{m}\sum_{k=2}^{m-1} \frac{(nq)^m}{n\cdot (nq)^k}\left(k(k+1) + O\left(\frac{k^4m^6}{n}\right)\right) \\
        &\le 2(m-2)n^{m-3} q^{m} 
        + 4(m-2)^2 n^{2m-6} q^{2m-1}
        + n^{m-3}q^{m}\sum_{k=2}^{m-1} \frac{(C\log n)^\xi}{n\cdot (C\log n)^k}\left(k(k+1) + O\left(\frac{k^4m^6}{n}\right)\right) \\
        &\le \left(2(m-2) + o(1)\right)\cdot n^{m-3} q^{m}.
    \end{align*}
    Putting these together, the correlation \pref{eq:cor} is $\ge 1-o(1)$, as desired.
    
    Clearly, $f_{ij}$ is a polynomial of degree $m= o(\log n)$ in the entries of the adjacency matrix of $\bG$.
    The polynomial can be computed approximately in polynomial time, either approximately computed using color-coding \cite{AYZ95,AR02}, or, alternatively, by breadth-first search in $G$, on the event that $G$ is sparse (more discussion in the proof of \pref{lem:sampling}). 
\end{proof}

\subsection{Sampling despite disorder chaos}\label{sec:disorder}
Disorder chaos is a stronger version of the overlap-gap property: it is a more quantitative overlap-gap statement about samples from the distribution over good solutions.
In our context, we would say that the shortest-paths problem exhibits disorder chaos if with high probability over $\rho$-correlated $\bG,\bG'\sim \bG(n,q)$,
\[
\lim_{\rho \to 1} \E_{\bp,\bp' \sim \mathrm{Unif}(\calP_\eps(\bG))\otimes \mathrm{Unif}(\calP_\eps(\bG'))}\left[\left(\frac{\langle\bp,\bp'\rangle}{\ell^*}\right)^2\right] = o_n(1),
\]
Where $\ell^* = \frac{\log n}{\log nq}$ is the length of a typical path.
One could instead consider a different measure over $\calP_\eps(\bG)$, but we take the uniform measure here for simplicity.

Disorder chaos sometimes implies lower bounds against stable algorithms for sampling from $\Unif(\calP_\eps(\bG))$ (see \cite{el2022sampling}).
The argument goes by showing that, under certain conditions, disorder chaos implies a lower bound on the Wasserstein distance between $\Unif(\calP_\eps(\bG))$ and $\Unif(\calP_\eps(\bG'))$.
Let $\Pi(\bG,\bG')$ be the space of couplings of $\Unif(\calP_\eps(\bG))$ and $\Unif(\calP_\eps(\bG'))$.
If one can show that with high probability,
\begin{equation}
\lim_{\rho \to 1} \inf_{\pi \in \Pi(\bG,\bG')} \E_{(\bp,\bp')\sim \pi}\left[\frac{1}{\ell^*}\|\bp - \bp'\|^2\right] > c - o_n(1)\label{eq:wass}
\end{equation}
for some universal constant $c > 0$, then a concise argument\footnote{The triangle inequality.} of \cite{el2022sampling} yields lower bounds against smooth algorithms for sampling from $\mathrm{Unif}(\calP_\eps(\bG))$ in the Wasserstein distance.

Here, we will establish that shortest paths in $\bbG(n,q)$ have the Wasserstein distance bound as in \pref{eq:wass}, implying that stable algorithms cannot sample.
Furthermore, we will show that there is an efficient, non-stable algorithm that, with high probability over $\bG \sim \bbG(n,q)$ succeeds in producing exact samples from $\Unif(\calP_\eps(\bG))$.

\paragraph{A lower bound on the Wasserstein distance.}
Here we establish that \pref{eq:wass} holds for shortest path in $\bbG(n,q)$, implying that stable algorithms cannot sample from $\Unif(\calP_\eps(\bG))$.
\begin{lemma}\label{lem:disorder}
    Suppose $\rho < 1$, and $\bG,\bG' \sim \bbG(n,q)$ are $\rho$-correlated, with $q = \frac{C\log n}{n}$ and $C \ge 1$.
    Let $\Pi(\bG,\bG')$ be the set of all couplings of $\Unif(\calP_\eps(\bG)),\Unif(\calP_\eps(\bG'))$.
    Then for all $\delta > 0$, there exists $n$ sufficiently large so that with probability $1-o_n(1)$,
    \[
    \inf_{\pi \in \Pi(\bG,\bG')} \E_{(\bp,\bp')\sim \pi}\left[\frac{1}{\ell^*}\|\bp - \bp'\|^2\right] \ge 1 - \delta.
    \]
\end{lemma}
\begin{proof}
    Suppose $p$ is a $1$-$2$ path of length $L \in [1,1+\eps]\ell^*$, where $\ell^* = \frac{\log n}{\log nq}$.
    We will argue that conditioned on $p \in E(\bG)$, with high probability, there is no $\bp' \in \calP_\eps(\bG')$ for which $\langle p,\bp'\rangle \ge \frac{1}{2}\ell^*$.
    This is enough to establish \pref{eq:wass} with $c = 1$, as for all but a vanishing measure of $\bp \sim \Unif(\calP_\eps(\bG))$, there is no $\bp' \in \calP_\eps(\bG')$ to which they could be transported with cost $< 1$.

    The argument is a first moment argument similar to the one used to establish the OGP.
    Conditioned on $p \in E(\bG)$, for any given $1$-$2$ path $p'$ of length $L' \in [1,1+\eps]\ell^*$ which overlaps with $p$ in $k$ edges,
    \[
    \Pr[p' \in \bG' \mid p \in \bG] = (\rho + (1-\rho)q)^k q^{L'-k}.
    \]
    By \pref{claim:path-pairs}, the number of such paths $p'$ is at most $\frac{1}{n^{\ul{L}-1}}M_{k,L-k,L'-k} \le (1+o_n(1)) n^{L'-k}\left(\frac{k+1}{n^2} + \frac{\log^{15} n}{n^3} \right)$ when $k < \min(L,L')$, and at most $1$ when $k = L = L'$.
    So, summing from $k= \frac{1}{2}\ell^*,\ldots, L'$,
    \begin{align*}
    &\E[\# p' \in \calP_\eps(\bG),\, \langle p',p \rangle \ge \frac{1}{2}\ell^*]\\
    &\le (\rho + (1-\rho)q)^L + \sum_{k=\frac{1}{2}\ell^*}^{L'-1} (\rho + (1-\rho)q)^k q^{L'-k}(1+o_n(1)) n^{L'-k}\left(\frac{k+1}{n^2} + \frac{\log^{15} n}{n^3} \right)\\
    &\le o_n(1) + O\left(\frac{(nq)^{L'}\log n}{n^2}\right) \sum_{k=\frac{1}{2}\ell^*}^{L'} \left(\frac{(1+o_n(1))\rho}{nq}\right)^k \\
    &= o_n(1) + O\left(\frac{\log n}{n^{1-\eps}}\right).
    \end{align*}
    Thus, by Markov's inequality, with high probability $p$ will have no $\bp' \in \calP_\eps(\bG')$ which overlaps by $\frac{1}{2}\ell^*$ or more, as desired.
\end{proof}

\paragraph{Polynomial-time algorithm for sampling.}
\begin{lemma}\label{lem:sampling}
There is a polynomial-time algorithm which with high probability succeeds in sampling from $\mathrm{Unif}(\calP_\eps(\bG))$ when $\bG \sim \bbG(n,q)$ for $q \ge \frac{\log n}{n}$.
\end{lemma}
\begin{remark}\label{rem:weights}
    We remark that the proof of \pref{lem:sampling} can easily be extended to efficiently sample from any distribution on $\calP_\eps(\bG)$ where the weight of each path is polynomial-time computable from its description.
\end{remark}
\begin{proof}[Proof of \pref{lem:sampling}]
    From \pref{lem:shortest-gnp}, the length of the shortest path concentrates around $\ell^* = \frac{\log n}{\log nq}$ with high probability.
    So the algorithm is as follows: one simply enumerates all walks of length $L = (1+\eps)\ell^*$ starting at vertex $1$, and then deletes from the list any walk that either is not simple, or did not encounter vertex $2$.
    Finally, one of these paths is chosen uniformly at random (or, if one wishes, one can add weights which depend on the path length, as discussed in \pref{rem:weights}).
    
    This can be implemented with a stack or a queue, in a manner similar to depth-first search or breadth-first search, except that one must keep track of the length of the path explored so far at every vertex, and the stopping criterion is based on the length of the path rather than on whether vertices have already been explored.
    If $\bd_{\max}$ is the maximum degree of a vertex in $\bG$, then this algorithm comprises of at most  $\bd_{\max}^L$ operations, each requiring the recording of $O(\log n)$ bits of information.
    And in the end, there are at most $\bd_{\max}^L$ paths to work with.
    So provided that $\bd_{\max}^L = \poly(n)$, the algorithm is polynomial time.

    We now argue that this holds with high probability.
    Indeed, the degree $\bd_i$ of vertex $i$ is distributed as $\Bin(n-1,q)$, and so applying Markov's inequality with the moment generating function,
    \[
    \Pr[\bd_i > 3q(n-1)] 
    \le \frac{\E e^{\bd_i}}{e^{3q(n-1)}}
    = \frac{(1-q+qe)^{n-1}}{e^{3q(n-1)}}
    \le \frac{e^{(e-1)q(n-1)}}{e^{3q(n-1)}}
    \le e^{(e-4) q(n-1)} \ll \frac{1}{n^{1.1}},
    \]
    as $q \ge \frac{\log n}{n}$. 
    Hence by a union bound, with high probability the maximum degree is no larger than $3nq$.
    Then with high probability, 
    \[
    \bd_{\max}^L \le (3nq)^{L} = (3nq)^{(1+\eps)\frac{\log n}{\log nq}} \le n^{1+\eps + o(1)},
    \]
    and we have our conclusion.
\end{proof}

\section{First passage percolation on the complete graph}

In this section, we establish an overlap-gap property for shortest paths in the complete graph with random non-negative edge weights.
This problem, too, is solved in polynomial time by the shortest path algorithm.

Our setting will be as follows:
Sample a vector $\bell \in \R_{+}^{\binom{[n]}{2}}$ with i.i.d. $\Exp(1)$ entries.
We will use $\bell$ as an assignment of edge lengths to the edges of $K_n$, and denote the resulting randomly weighted graph by $K_n^{\bell}$.
For each path $p$ in $K_n$, we identify $p$ with the vector in $\{0,1\}^{\binom{[n]}{2}}$ whose entries indicate membership in $p$.
We define the length of $p$ in $K_n^{\bell}$ to be $\len_{\bell}(p) = \langle p,\bell\rangle$, and we define the \emph{hopcount} of $p$ to be $\hop(p) = \|p\|^2$.

In this case our object of study will be minimum-length paths between vertices $1$ and $2$ in $K_n^{\bell}$ (note that in $\bbG(n,q)$ the hopcount and length coincide).
We will begin by listing (and deriving) some properties of the approximate min-length paths, describe a natural distribution over correlated instances, derive the overlap-gap property and lower bounds against smooth algorithms, and finally comment on the existence of efficient algorithms for sampling from $\calP_\eps(\bell)$ and the possibility of a low-degree polynomial estimator.

\subsection{Characterizing approximate shortest paths in weighted complete graphs}
The properties of \emph{the} minimum-length path are studied under the name ``first passage percolation'' on the complete graph.
The behavior of the length and hopcount of the shortest path is well-understood \cite{BGRS04}, as summarized in the following theorem (Theorem 1.1 and 1.2 in \cite{BvdH}):
\begin{theorem*}
   If $\bp$ is the shortest path between $1,2$ in $K_n^{\bell}$, then
   \[
        \left(n \cdot\len_{\bell}(\bp) - \log n, \frac{\hop(\bp)-\log n}{\sqrt{\log n}} \right) \xrightarrow[]{d} (\bX,\bZ),
   \]
   for independent random variables $\bZ,\bX$.\footnote{Further, $\bZ$ is a standard Normal random variable. The distribution of $\bX$ is more involved.}
\end{theorem*}

Here, we are interested in the set of $\eps$-approximate shortest paths, 
\[
\calP_\eps(\bell) \defeq \left\{ p \mid \len_{\bell}(p) \le (1+\eps)\tfrac{\log n}{n}\right\}.
\]
For paths $p,p'$, we define their overlap in terms of the hopcount overlap $\frac{\langle p,p'\rangle}{\|p\|\cdot\|p'\|}$.
To demonstrate the overlap-gap property for $\calP_\eps(\bell)$, it will be helpful to have a bound on the hopcount of an approximate shortest path.
\begin{lemma}\label{lem:path-range}
Let $\eps \in (0,1)$ be a sufficiently small constant.
If $p \in \calP_\eps(\bell)$, then with probability $1-o_n(1)$,
 $\left|\frac{\hop(p)}{\log n} - 1\right|\le 2\sqrt{\eps}$.
\end{lemma}
\begin{proof}
Let $\bX^{(k)}$ be the number of paths of hopcount $k$ and length at most $L = (1+\eps)\tfrac{\log n}{n}$ between $1$ and $2$.
The number of paths $p$ of hopcount $k$ between $1,2$ is precisely $(n-2)^{\ul{k-1}}$, and the length $\bX_p$ of each such path is distributed as the sum of $k$ independent $\Exp(1)$ random variables, that is, $\bX_p$ is $\mathrm{Gamma}(k,1)$-distributed.
Hence, integrating the density function directly,
\[
\Pr[\bX_p \le L] = \int_0^L \frac{e^{-a}a^{k-1}}{(k-1)!} da \le \frac{1}{k!} L^k.
\]
Thus, combining the above and applying Stirling's formula,
\[
\E[\bX^{(k)}] \le (n-2)^{\ul{k-1}} \cdot \Pr[\bX_p \le L] \le n^{\ul{k-1}} \frac{1}{k!} \left(\frac{(1+\eps)\log n}{n}\right)^k \le \left(\frac{e(1+\eps)\log n}{k}\right)^{k-1} \cdot \frac{\log n }{n}.
\]

In the case that $k > e^2 (1+\eps) \log n$, $\E[X^{(k)}] \le \frac{1}{n^2}$, so union bounding over the at most $n$ choices of such $k$, we have by Markov's inequality that the probability that any $k$-hop path with $k > e^2 (1+\eps)\log n$ is $(1+\eps)$-approximately minimum:
length is at most $1/n$.

Henceforth, we may focus on $k$ of order $\log n$.
Writing $k =\beta \log n$, we have
\[
\log(\E[\bX^{(k)}]) \le \left(\beta + \log(1+\eps)\beta  - \beta\log \beta - 1\right)\cdot \log n  + \log\log n.
\]
This function is concave in $\beta$, maximized at $\beta = (1+\eps)$.
Hence it suffices to verify that the value is negative at $\beta = (1-2\sqrt{\eps})$ and $\beta = (1+2\sqrt{\eps})$, when $\eps$ is sufficiently small.
Indeed, if $\beta =1+\delta$ with $|\delta| < 1$, using the series expansion for $\log(1+\delta)$,
\begin{align*}
1+\delta + (1+\delta)\log(1+\eps) - (1+\delta)\log(1+\delta) - 1
&\le \delta + \eps + \eps\delta - (1+\delta)(\delta - \frac{\delta^2}{2}+\frac{\delta^3}{3} - \cdots)\\
&= \eps + \eps\delta - \frac{1}{2}\delta^2 + O(\delta^3).
\end{align*}
For $|\delta| = 2\sqrt{\eps}$ and $\eps$ sufficiently small, this is negative.
Thus, for $k \le (1-2\sqrt{\eps}) \log n$ and $(1+2\sqrt{\eps})\log n \le k \le e^2(1+\eps)\log n$, $\E[\bX^{(k)}]\le n^{-c}$ for $c>0$ a constant depending only on $\eps$.
Applying Markov's inequality and taking a union bound over all $O(\log n)$ such $k$ completes the proof.
\end{proof}

\subsection{Correlated instances}
In order to state our overlap-gap property for correlated instances, we introduce a canonical joint distribution over sets of correlated edge weights $\bell$.

Given an exponential random variable $\bX_0 \sim \Exp(1)$ and a time $t \in \R_+$, running the following Langevin diffusion process for time $t$ produces an exponential random variable $\bX_t \sim \Exp(1)$, by taking:
\[
d\bX_s = -\Ind[\bX_s \ge 0] ds + \sqrt{2}d\bB_s
\]
for $(\bB_s)_{s\ge 0}$ a standard Brownian motion.
For a function $f:\R_+\to\R$, we use the notation
\[
P_t f(x) = \E[f(\bX_t)\mid \bX_0 = x].
\]

We consider correlated weights $\bell^0$, $\bell^t$ given by sampling $\bell^0 \sim (\Exp(1))^{\otimes \binom{[n]}{2}}$, and then applying the time-$t$ Langevin process to each coordinate $e \in \binom{[n]}{2}$ independently to produce $\bell^t$.
One can compute that
\[
\mathrm{Corr}(\bell^0_e,\bell^t_e) = \frac{\E[(\bell^0_e-1)(\bell^t_e-1)]}{\Var{\bell_e}} = e^{-t}.
\]

\begin{lemma}\label{lem:correlation}
    Let $\calF,\calG\subset \R^{\binom{[n]}{2}}$ be two events. 
    Let $(\bell^s)_{s \ge 0}$ be vectors sampled as described above.
Then for any $t \ge 0$,
    \[
    \Pr[\calF(\bell^0),\calG(\bell^t)] \le (1-e^t)\Pr[\calF(\bell^0)]\Pr[\calG(\bell^0)] + e^{-t}\Pr[\calF(\bell^0),\calG(\bell^0)].
    \]
\end{lemma}
\begin{proof}
Let $f(\ell) = \Ind[\calF(\ell)]$ and let $g(\ell) = \Ind[\calG(\ell)]$.
As is explained in e.g. \cite{BGL} (see section 2.7), there is a basis of polynomials, the \emph{Laguerre polynomials}, $\{L_k\}_{k \in \N}$, which form an orthonormal basis for functions in $L^2$ with respect to the exponential measure, and such that for each $e \in \binom{[n]}{2}$, 
\[
\E[L_k(\bell^0_e)P_t L_m(\bell^0_e)] = \E[L_k(\bell^0_e)L_m(\bell^t_e)] = \Ind[k=m]e^{-tk}.
\]
Moreover the variables corresponding to each edge are independent. 
Therefore we may expand
$f(\ell) = \sum_{\alpha \in \N^{\binom{[n]}{2}}} \hat f_\alpha \cdot \prod_{e \in \binom{[n]}{2}} L_{\alpha(e)}(\ell_e)$, noting that $\hat f_0 = \E[f(\bell)] = \Pr[\calF(\bell^0)]$ and $\Var(f) = \sum_{|\alpha|\ge 1} \hat f_\alpha^2$, and similarly for $g(\ell)$.
Then we may compute 
\begin{align*}
\Pr[\calF(\bell^0),\calG(\bell^t)]
&= \E[f(\bell^0)g(\bell^t)]\\
&= \sum_{\alpha,\beta \in \N^{\binom{[n]}{[2]}}} \hat f_\alpha \hat g_\beta \prod_{e \in \binom{[n]}{2}} \E[L_{\alpha(e)}(\bell^0_e)L_{\beta(e)}(\bell^t_e)]\\
&= \sum_\alpha \hat f_\alpha \hat g_\alpha e^{-t|\alpha|}
\le \E[f]\E[g] + e^{-t}\left(\E[f(\bell^0)g(\bell^0)] - \E[f]\E[g]\right).
\end{align*}
The conclusion now follows because $f$ and $g$ are indicators.
\end{proof}

We note that the Langevin dynamics evolves the weights smoothly:
\begin{claim}\label{claim:smooth-langevin}
For any fixed $t > 0$, if $\bell^0$ and $\bell^t$ are sampled as above, then for all $n$ large enough,
\[
\|\bell^0 - \bell^t\|_2 \le \sqrt{2(1-e^{-t}) + o_n(1)}\|\bell^0 - \E[\bell^0]\|_2 \le \sqrt{1-e^{-t} + o_n(1)}\|\bell^0\|_2.
\] 
with probability $\ge 1-\frac{1}{\poly n}$.
\end{claim}
\begin{proof}
As mentioned above, the correlation between $\bX_0$ and $\bX_t$ can be computed as $e^{-t}$.
The claim now follows from the classic equality $\E[(\bX - \bX')^2] = 2\Var[\bX] - 2\Cov[\bX,\bX']$ for $\bX,\bX'$ with the same marginal distribution, combined with concentration of sums of independent random variables, and finally the fact that $\Var(\bX) = 1$ when $\bX \sim \Exp(1)$.
\end{proof}

\subsection{The overlap-gap property for first passage percolation}

We now show that any two paths $p_1 \neq p_2$ in $\calP_\eps(\bell^0)\times \calP_{\eps}(\bell^t)$ have the overlap-gap property, and that when $t$ is large enough large overlap is improbable.

\begin{lemma}\label{lem:ogp-fpp-single}
For any $t \ge 0$, let $\bell^0 \sim \Exp(1)^{\otimes\binom{[n]}{2}}$ and let $\bell^t$ be the outcome of a time-$t$ Langevin diffusion starting from $\bell^0$.
There exist universal constants $C_1,C_2>0$ and $\eps_0>0$ such that for all $\eps < \eps_0$, it holds with probability $1-1/\poly(n)$ that all pairs $(p,p') \in \calP_\eps(\bell^0)\times \calP_\eps(\bell^t)$ have
\[
\frac{|p\cap p'|}{|p|\cdot|p'|} \in \begin{cases}
 [0, C_1\sqrt{\eps})& t > (1+C_2\sqrt{\eps})\log n\\
 [0, C_1\sqrt{\eps}) \cup \{1\} & \text{otherwise}.
\end{cases}
\] 
Furthermore, with probability $1-O(\frac{1}{\log n})$ there exist $p,p'\in \calP_\eps(\bell^0) \times \calP_\eps(\bell^0)$ with overlap $\le C_1\sqrt{\eps}$.
\end{lemma}

\begin{proof}
    To establish the overlap-gap property, we bound the expected number of overlapping paths.
	We begin by giving a bound for non-full overlaps in the case when $t = 0$; for simplicity's sake we call $\bell = \bell^0$. 
By \pref{lem:correlation}, the bound will transfer to any $t \ge 0$.
We will then get a sharper bound on the full-overlap case when $t > 0$.

    Suppose paths $p_1,p_2$ (both with endpoints $1$,$2$) overlap on $k$ edges, and that $\hop(p_1) = k + m_1$ and $\hop(p_2) = k + m_2$, with $m_1,m_2 \neq 0$.
    We will bound the expected number of pairs of such paths with $\len_{\bell}(p_1),\len_{\bell}(p_2) \le (1+\eps)\frac{\log n}{n}$.

    First, we compute the probability that $\len_{\bell}(p_1),\len_{\bell}(p_2) \le (1+\eps)\frac{\log n}{n}$.
    The length of path $p_1$ is the sum of the lengths of the $k$ shared edges, plus the sum of the lengths of the $m_1$ edges unique to $p_1$ (and the same for $p_2$).
    So letting $L = (1+\eps)\frac{\log n}{n}$ and defining the independent random variables $\bA \sim \mathrm{Gamma}(k,1)$, $\bB \sim \mathrm{Gamma}(m_1,1)$, and $\bC \sim \mathrm{Gamma}(m_2,1)$, we have by direct integration of the density functions that
    \begin{align}
       \Pr[\len_{\bell}(p_1),\len_{\bell}(p_2) \le L] 
       &= \Pr[\bA + \bB \le L, \, \bA + \bC \le L]\nonumber\\
       &= \int_0^{L} \frac{a^{k-1}e^{-a}}{(k-1)!}\left(\int_0^{L -a} \frac{b^{m_1-1} e^{-b}}{(m_1-1)!}db \right)\left(\int_0^{L-a}\frac{c^{m_2-1}e^{-c}}{(m_2-1)!} dc\right) da \nonumber \\
       &\le \frac{1}{(k-1)!(m_1-1)!(m_2-1)!}\int_0^L a^{k-1} \int_0^{L-a} b^{m_1-1} db \int_0^{L-a} c^{m_2-1}dc\, da \nonumber \\
       &= \frac{1}{(k-1)!m_1!m_2!}\int_0^L a^{k-1} (L-a)^{m_1+m_2} da\nonumber \\
       &= \frac{1}{(k-1)!m_1!m_2!}\frac{(k-1)!(m_1+m_2)!}{(m_1+m_2+k)!} L^{k + m_1 + m_2} \label{eq:pairbar2}\\
       &\le \binom{m_1+m_2}{m_1} \cdot \left(\frac{e L}{k+m_1+m_2}\right)^{k+m_1+m_2},\nonumber
    \end{align}
    where in the last line we have applied Stirling's approximation.

Applying \pref{claim:path-pairs}, so long as $m_1,m_2 \neq 0$, the expected number of pairs of approximate shortest paths, in this case, is at most
    \begin{align*}
    \E[\#\{p_1,p_2\}] 
    &= M_{k,m_1-k,m_2-k} \cdot \Pr[\len_{\bell}(p_1),\len_{\bell}(p_2) \le (1+\eps)\tfrac{\log n}{n}]\nonumber\\
    &\le \left(\frac{k+1}{n^2} + \frac{O(\log^{15} n)}{n^3}\right) \cdot n^{k+m_1+m_2} \cdot\binom{m_1+m_2}{m_1} \cdot \left(\frac{e(1+\eps)\log n}{(k+m_1+m_2) n}\right)^{k+m_1+m_2}.
    \end{align*}
    
    Without loss of generality, from \pref{lem:path-range}, we may assume that the hoplength of $p_1$ and $p_2$ is $(1\pm 2\sqrt{\eps})\log n$, and we will do so from now on to simplify the computation. 
    Suppose $k = (1-\lambda) \log n$, $m_1 = \left(\lambda + \delta_1\right)\log n$, and $m_2 = \left(\lambda + \delta_2\right)\log n$ for $|\delta_1|,|\delta_2| \le 2\sqrt{\eps}$.
    Bounding the display above, for $\eps$ sufficiently small,
    \begin{align*}
    \E[\#\{p_1,p_2\}] &\le
    \left(\frac{O(\log n)}{n^2}\right) \left(\frac{e\left(1+\eps\right)}{1+\lambda -4\sqrt{\eps} }\right)^{(1+\lambda + 4\sqrt{\eps})\log n} \cdot 2^{(2\lambda +4\sqrt{\eps})\log n}.
    \end{align*}
    Taking logarithms,
    \begin{align*}
    \frac{\log \E[\#\{p_1,p_2\}]}{\log n} 
    &\le - 2 + (1+\lambda+4\sqrt{\eps})\log\left(\frac{e(1+\eps)}{1+\lambda-4\sqrt{\eps}}\right) + (2\lambda +4\sqrt{\eps})\log 2 + o_n(1)\\
    &\le -2 + C \sqrt{\eps} + (1+\lambda)\log\frac{e}{(1+\lambda)}+ 2\lambda\log 2  + o_n(1),
    \end{align*}
    for $C$ a universal constant, whenever $\eps$ is sufficiently small. 
    Taking a first derivative in $\lambda$, one can see that the above expression is increasing in $\lambda$ when $\lambda \in (-1,3)$.
    Further, evaluating the expression at $\lambda = 1-\beta$ for $\beta < \frac{1}{2}$, we obtain an upper bound:
    \begin{align*}
    &-2 + C \sqrt{\eps} + (2-\beta)\left(1 - \log(2-\beta)\right) + 2(1-\beta)\log 2\\
    &= C \sqrt{\eps} - \beta + 2 \log \frac{2}{2-\beta} + \beta \log(2-\beta) - 2\beta \log 2\\
    \intertext{Applying the identity $\log \frac{1}{1-x} \le x + x^2$ valid for all $0 \le x \le \frac{1}{2}$,}
    &\le C \sqrt{\eps} - \beta + (\beta + \frac{1}{2}\beta^2) + \left(\beta \log 2 +\beta \log(1-\beta/2) \right)- 2\beta\log 2\\
    &\le C \sqrt{\eps} + \frac{1}{2}\beta^2 - \beta \log 2 +\beta \log(1-\beta/2) \\
    \intertext{And finally applying the identity $\log(1-x) \le -x$,}
    &\le C \sqrt{\eps} - \beta \log 2.
    \end{align*}
    Hence the quantity is negative when $\beta \ge C\sqrt{\eps}/2\log 2$.
    Combined with the discussion above, when $-1 < \lambda < 1 - C \frac{\sqrt{\eps}}{\log 2}$ and $(1-\lambda)\log n < |p_1|,|p_2|$, $\E[\#\{p_1,p_2\}] \le n^{-\gamma}$ for $\gamma > 0$ a universal constant, for all $n$ sufficiently large.
    Applying Markov's inequality and taking a union bound over the $O(\log^4 n)$ possible values of $c,m_1,m_2,k$, we have that with high probability, there do not exist $p_1,p_2 \in \calP_\eps(\bell)$ of overlap $\alpha = \frac{1-\lambda}{\sqrt{(1+\delta_1)(1+\delta_2)}} \in \left(C_1 \sqrt{\eps}, 1\right)$.
The same is true if $t > 0$, only the probability term changes, and by \pref{lem:correlation} it can only decrease.
\smallskip

Now, in the case when $m_1=m_2=0$, we consider $t > 0$.
Using again our bound on the hoplength from \pref{lem:path-range}, $k = (1\pm 2\sqrt{\eps})\log n$, and we have that
\[
\E[\#{p_1,p_2}] 
\le n^k \left( e^{-t} \left(\frac{e(1+\eps)\log n}{nk}\right)^{2k}+ \left(\frac{e(1+\eps)\log n}{nk}\right)^{2k}\right)
\le e^{-t} \left(\frac{e(1+\eps)}{1-2\sqrt{\eps}}\right)^{(1+2\sqrt{\eps})\log n} +\frac{1}{\poly(n)},
\]
which is $o_n(1)$ for all $\eps$ sufficiently small when $t \ge (1+C_2 \sqrt{\eps})\log n$ for some constant $C_2$. 
The conclusion follows by Markov's inequality.

\medskip
To obtain the guarantee for the existence of near-disjoint paths, we argue that $|\calP_\eps(\bell^0)|\gg 1$ with high probability via a second moment argument.
Given that paths must overlap either fully or by at most $C_1\sqrt{\eps}$, this implies there must exist two paths of low overlap, by the pigeonhole principle.

Let $\bN_{h,L}$ be the number of paths of hoplength $h = \log n$  and cost at most $L = (1+\eps)\frac{\log n}{n}$.
Notice that $|\calP_\eps(\bell)| \ge \bN_{h,L}$.
We will show that $\bN_{h,L}$ concentrates around its expectation, which is at least 
\[
\E[\bN_{h,L}] 
= n^{\ul{h-1}}\int_0^L \frac{1}{(h-1)!}a^{h-1}e^{-a} da 
\ge n^{h-1} e^{-L} \frac{L^h}{h!}(1-o(1)) 
\ge \frac{1}{n}(e(1+\eps))^{\log n}(1-o(1)) \gg 1.
\]
The line \pref{eq:pairbar2} above, implies that
\begin{align*}
\E[\bN_{h,L}^2]-\E[\bN_{h,L}]^2 
&\le \E[\bN_{h,L}] + \sum_{k=1}^{h-1} O\left(\frac{k}{n^2}\right)n^{2h - k} \binom{2h-2k}{h-k} \left( \frac{(1+\eps)\log n}{n}\right)^{2h-k}\frac{1}{(2h-k)!}\\
&\le \E[\bN_{h,L}] + \sum_{k=1}^{h-1} O\left(\frac{k}{n^2}\right) \binom{2h-2k}{h-k} \frac{1}{(2h-k)!}\left( (1+\eps)\log n \right)^{2h-k}.
\end{align*}
If we define $A_k = k \binom{2h-2k}{h-k}/(2h-k)!$, then we have that $A_{k+1}/A_{k} = \frac{(k+1)(h-k)}{2k(2h-2k-1)(2h-k)}\le 3/4$ for any $1\leq k \leq h-1$, since $(k+1)/(2k) \le 3/4$ for any $k \geq 2$ and the inequality holds trivially when $k=1$. Therefore, $\sum_{k=1}^{h-1} A_k \leq 4 A_1$ and thus
\begin{align*}
\E[\bN_{h,L}^2]-\E[\bN_{h,L}]^2 
&\le \E[\bN_{h,L}] + O\left(\frac{1}{n^2}\right) \binom{2h-2}{h-1} \frac{1}{(2h-1)!}\left( (1+\eps)\log n \right)^{2h-1}\\
&\le \E[\bN_{k,L}] + O\left(\frac{1}{h n^2}\right) \frac{1}{(h-1)!(h-1)!}\left( (1+\eps)\log n \right)^{2h-1}\\
&\le \E[\bN_{h,L}] + O\left(\frac{h}{n^2}\right) \frac{1}{h!h!}\frac{1}{\log n}\left( (1+\eps)\log n \right)^{2h}\\
&\le \E[\bN_{h,L}] + O\left(\frac{h}{n^2}\right) \frac{1}{\log n}\frac{1}{\log n}\left( e(1+\eps) \right)^{2h}\\
&\le \E[\bN_{h,L}] + O\left(\frac{1}{n^2\log n}\right) \left( e(1+\eps) \right)^{2\log n} = O\left( \frac{1}{\log n}\right) \E[\bN_{h,L}]^2 .
\end{align*}
Hence, $\Var[N_{h,L}] = O(\frac{1}{\log n})\E[N_{h,L}]^2$, so $\bN_{h,L} \ge \frac{1}{2}\E[\bN_{h,L}] \gg 1$ with probability $1-O(\frac{1}{\log n})$, as desired.
\end{proof}
We can leverage the above OGP to prove a lower bound for stable algorithms.

\begin{lemma}\label{lem:fpp-nostable}
    Let $t>0$, $\eps > 0$ be sufficiently small.
    Then for all $n$ sufficiently large, there can be no algorithm for $\eps$-approximate shortest path in $K_n(\bell)$ which simultaneously (i) has success probability $\ge 1-o(\frac{1}{\log n})$, and (ii) is stable, in the sense that if $\bell^0,\bell^t$ are sampled as described above (and hence $\|\bell^0 - \bell^t\| \le \sqrt{(1-e^{-t}) + o_n(1)} \|\bell^0\|$ with high probability), then when $\calA$ succeeds on its inputs, $\frac{|\calA(\bell^0) \cap \calA(\bell^t)|}{\sqrt{|\calA(\bell^0)|\cdot|\calA(\bell^t)|}} > C\sqrt{\eps}$.
\end{lemma}
\begin{proof}
Sample a sequence $(\bell_s)_{s\ge 0}$ by sampling $\bell_0 \sim \Exp(1)^{\otimes \binom{[n]}{2}}$ and then running the Langevin dynamics described above.
Let $\bp_s = \calA(\bell_s)$.
Choose $K = \lceil \frac{2\log n}{t}\rceil$.
From \pref{thm:ogp-gnp}, we have that with probability at least $1-O(\frac{K}{\poly n})$, for each $k \in [K]$, the path $\bp_{tk}$ must have overlap with $\bp_0$ which is either equal to $1$, or at most $C\sqrt{\eps}$.
Furthermore, since $tK \ge 2\log n$, the overlap of $\bp_0$ and $\bp_T$  must be at most $C\sqrt{\eps}$.
Hence, with high probability there must exist some $k \in [K]$ with 
$\frac{|\bp_0 \cap \bp_{tk}|}{\sqrt{|\bp_0||\bp_{tk}|}} = 1$ but
$\frac{|\bp_0 \cap \bp_{t(k+1)}|}{\sqrt{|\bp_0||\bp_{t(k+1)}|}}\le C\sqrt{\eps}$, 
implying that $\frac{|\bp_{tk} \cap \bp_{t(k+1)}|}{\sqrt{|\bp_{tk}||\bp_{t(k+1)}|}} \le C \sqrt{\eps} $.
This is a contradiction, unless $\calA$ did not succeed on some input.
By a union bound $\calA$ must have been successful on all $K+1$ inputs with probability at least $1-\delta(K+1)$, so it must be the case that $\delta > \frac{1}{K+1} = \Theta(\frac{1}{\log n})$.
\end{proof}

\subsection{Low-degree polynomial estimators and sampling}

\paragraph{Bounded-degree polynomial estimators.}
It is not clear whether a degree-$O(\log n)$ polynomial estimator of the first passage percolation path exists.
A first attempt may be as follows: to estimate whether $(i,j)$ participates in the path, sum over all $1$-$2$ paths of hoplength $(1+2\sqrt{\eps})\log n$ containing $(i,j)$, a low-degree approximation for the indicator that the total length of the same path is at most $\frac{(1+\eps)\log n}{n}$.
However, the degree of the polynomial approximator would have to be $\Omega(n/\log n)$ to have the decision boundary on the order of $\log n / n$.

It is perhaps possible to have an estimator of low \emph{coordinate degree} \cite{hopkins2018statistical,brennan2021statistical,kunisky2024low}, meaning that the estimator is a linear combination of functions each depending on $O(\log n)$ coordinates of $\bell$.
For example, if we know the law $W$ of the edge weights that participate in the first passage percolation path, one can choose some parameter $k$, obtain $k$ even-measure quantile values $Q_1,\ldots,Q_k$ of $W$, and attempt the estimate 
\[
f_{ij}(\bell) = \sum_{\|p\|^2 \le (1+2\sqrt{\eps})\log n} \sum_{t \vdash_k E(p)} \prod_{(i,j) \in p} \Ind[\bell_{ij} \le Q_{t(i,j)}],
\]
where the notation $t \vdash_k E(p)$ means that we partition $E(p)$ into $k$ parts of equal size.

The estimator is of coordinate degree $O(\log n)$, and so long as $k=O(1)$, the number of terms is $2^{O(\log^2 n)}$.
We find it plausible that this estimator may work, but we will not undertake its analysis here.

\paragraph{Sampling.}
We remark that the sampling algorithm described in \pref{lem:sampling} will also work for sampling short paths in $K_n^{\bell}$, given the following observation: all edges participating in paths of length at most $(1+\eps)\frac{\log n}{n}$ must have weight at most $(1+\eps)\frac{\log n}{n}$. 
Therefore, we may sparsify $K_n$ by first deleting all edges of weight $> (1+\eps)\frac{\log n}{n}$.
The support of the resulting graph is the same as the support of $\bbG(n,q)$ for $q = \Theta(\frac{\log n}{n})$, therefore the algorithm described in \pref{lem:sampling} can be used to produce a list of all potential length $\le (1+\eps)\frac{\log n}{n}$ paths from $1$ to $2$, which will have polynomial size with high probability.
The list can further be filtered to remove paths that are too long.

We believe that a Wasserstein lower bound in the style of \pref{eq:wass} will rule out smooth sampling algorithms in this case too; however we did not undertake the (at this point, somewhat tedious) verification of this fact.

\section{Concentration of symmetric low-degree polynomials of sparse random variables}\label{sec:concentration}

In this section, we show the concentration of symmetric low-degree polynomials of sparse random variables, which is needed for the proof of \pref{thm:smooth}. 
We will appeal to a Gaussian invariance principle for correlated random variables proven in \cite{caravenna2023critical} via an interpolation argument, along with the result for Gaussian polynomials in \cite{GJW24} (which is a straightforward argument based on hypercontractivity).
Our main goal will be to establish bounds on the expected fourth moment of the gradient of a function associated with our polynomial.

Suppose that $f:\{0,1\}^{\binom{[n]}{2}} \to \R^{\binom{[n]}{2}}$ is a degree-$D$ polynomial approximation to the shortest path, symmetric under the action of $S_{[n]\setminus\{1,2\}}$.
Suppose $\bG,\bG'$ are $\rho$-correlated samples from $\bbG(n,q)$ with $\rho \ge 0$.
With probability $\ge (1-q)^2$, the edge $(1,2)$ is in neither of $\bG,\bG'$.
Henceforth, we condition on the event that $(1,2)$ is in neither $\bG,\bG'$.

Let $\calH$ be the set of all labeled graphs $H = (V(H),E(H))$, with up to $2$ special vertices labeled $s(H), t(H)$, and $|E(H)|\le D$.
Let $\hookrightarrow$ denote an injection, and for $H \in \calH$ let $\calL(H)$ denote the set of all labelings $\ell:V(H) \hookrightarrow [n]$ such that $\ell(s(H)) = 1$ and $\ell(t(H)) = 2$.

We say a graph $H$ is of uniform edge density at most $\alpha$ if all subgraphs $J \subseteq H$ have edge density at most $\alpha$.
    Note that if $H$ is of uniform edge density at most $1$, each of its connected components contains at most one cycle.

\begin{definition}
    Let the \emph{active graph shapes} be $\calH$, the set of all graphs $H$ satisfying the following conditions: (1) $|E(H)| \le D$, (2) the vertices of $H$ are unlabeled save for at most four ``special'' vertices labeled $s,t,u,v$, (3) $H$ is of uniform edge density at most $1$, and (4) any component of $H$ containing the special vertices $s$ or $t$ can contain at most one of them, and is a tree.
\end{definition}

    \begin{lemma}
       If $f: \{0,1\}^{\binom{[n]}{2}} \to \R^N$ is a vector-valued polynomial of degree $D = O(1)$, then there exists a constant $\delta > 0$ depending only on $D$ such that so long as $q_n \le \frac{1}{n^{1-\delta}}$, with probability $1-o_n(1)$ over $\bG \sim \bbG(n,q_n)$, 
       \[
       f(\bG) = f^{\calH}(\bG),
       \]
       where $f^{\calH}$ is the restriction of $f$ to monomials of the form $\prod_{e \in E(H)} \Ind[\bG_e =1]$ for subgraphs $H$ which are isomorphic to some active graph shape $H^* \in \calH$, when the special label $s$ is identified with $1$ and the special label $t$ is identified with $2$.
    \end{lemma}
    \begin{proof}
        For any $H$, $\prod_{e \in E(H)} \Ind[\bG_e = 1]$ is nonzero only when, for all $J \subseteq H$, $\prod_{e \in E(J)} \Ind[\bG_e=1]$ is nonzero.
        So the statement holds if we can show that with probability $1-o_n(1)$, $\bG$ contains no subgraphs $J$ with at most $D$ edges which either (a) have density $> 1$, (b) is connected and contains both vertex $1$ and $2$, or (c) contains one of $1$ or $2$, is connected, and is denser than a tree.
        The argument will be via the first moment method in all three cases.
        \medskip 

        \noindent Case (a): if $J$ has $k$ vertices and $m > k$ edges for $m \le D$, then the expected number of appearances of $J$ is $\le n^{k} q_n^{m} \le n^{k}q_{n}^{k+1} \le \frac{1}{n^{1 - \delta(k+1)}} = o_n(1)$ if $\delta$ is small enough as a function of $D$.
        \medskip
        
        \noindent Case (b): if $J$ has $k$ vertices and $m \ge k-1$ edges (the minimum necessary for it to be connected) for $m \le D$, and two of the vertices are required to map to $1$ and $2$, then the expected number of appearances of $J$ is $\le n^{k-2} q_n^m \le n^{k-2} q_n^{k-1} \le \frac{1}{n^{1-\delta(k-1)}} = o_n(1)$ if $\delta$ is small enough as a function of $D$.

        \medskip
        \noindent Case (c): if $J$ has $k$ vertices and $m \ge k$ edges (the minimum necessary for it to be connected and have a cycle), and at least one of the vertices must map to $1$ or $2$, then the expected number of appearances of $J$ is $\le n^{k-1} \cdot 2 \cdot q_n^{k} \le 2 \frac{1}{n^{1-\delta(k)}} = o_n(1)$ (again when $\delta$ is small enough as a function of $D$.
        
        Taking a union bound over the at most $\exp(D^2)$ graphs on $D$ vertices completes the proof.
    \end{proof}

Henceforth, we shall assume $f$ is supported on active graphs, and we write $f = f^{\calH}$ (we will drop the superscript $\calH$ to keep the notation clean).

We let $s_H = |V(H)\cap \{s,t\}|$, $u_H = |V(H) \cap \{u,v\}|$, $v_H = |V(H)\setminus\{s,t,u,v\}|$, and $e_H = |E(H)|$. 
For each $H \in \calH$ and $i\neq j \in [n]$, let $\calL_{H,i,j}$ be the set of one-to-one maps $\ell:V(H) \hookrightarrow [n]$ which satisfy that $\ell(s) = 1,\ell(t) = 2,\ell(u) = i,\ell(v) = j$; each map appears with multiplicity one, meaning that if $\ell(H)$ and $\ell'(H)$ are isomorphic then $\ell = \ell'$.

When $f$ is symmetric under the action of $S_{[n]\setminus\{1,2\}}$, we can always write 
\[
f_{ij}(\bG) = \sum_{H \in \calH} \hat f_{H} \sum_{\ell \in \calL_{H,i,j}} \chi_{\ell(H)}(\bG),
\]
where the functions $\chi_\alpha$ are the orthonormal basis of Walsh-Hadamard characters, with $\chi_\alpha(G) = \prod_{(i,j) \in \alpha} \frac{G_{ij}-q}{\sqrt{q(1-q)}}$.

We assume $f$ is normalized so that $\E[\|f(\bG)\|^2] =1$, which means that
\begin{align*}
    1 &= \sum_{ij} \E \left[\left(\sum_{H \in \calH} \hat f_H \sum_{\ell \in \calL_{H,i,j}} \chi_{\ell(H)}(\bG)\right)^2\right]\\
    &= \sum_{ij} \sum_{H \in \calH} \hat f_H^2 \cdot \frac{n^{\ul{v_H}}}{\aut(H)}\\
    &= \sum_{H \in \calH} \hat f_H^2 \cdot \frac{n^{\ul{v_H} + \ul{u_H}}}{\aut(H)}.
\end{align*}
where $\aut(H)$ is the number of automorphisms of $H$.
Thus $|\hat f_H| \lesssim \left(\frac{1}{n}\right)^{\frac{1}{2}(v_H + u_{H})}$ when $D=O(1)$.

We now compare the polynomial of $\bG$ with that of correlated Gaussians. Recall that $\bG,\bG'$ are $\rho$-correlated samples from $\bbG(n,q)$. The inner product $c(\bG,\bG') = \langle f(\bG),f(\bG') \rangle$ can be written as 
\[
c(\bG,\bG') =  \sum_{i,j} \sum_{H_1,H_2 \in \calH} \hat f_{H_1} \hat f_{H_2}\sum_{\ell_1 \in \calL_{H_1,i,j}, \ell_2 \in \calL_{H_2,i,j}} \chi_{\ell_1(H_1)}(\bG) \chi_{\ell_2(H_2)}(\bG').
\]
Note that for each edge $a$, if we denote $\bchi_a^{(1)} = \chi_a(\bG)$ and $\bchi_a^{(2)} = \chi_a(\bG')$, then $c(\bG,\bG') = \Phi(\bchi^{(1)},\bchi^{(2)})$, where $\Phi(\cdot)$ is a multi-linear polynomial on $2\binom{n}{2}$ variables $\bchi^{(1)}$ and $\bchi^{(2)}$. Further,
\begin{align*}
    \|f(\bG)-f(\bG')\|_2^2 = c(\bG,\bG)+c(\bG',\bG')-2c(\bG,\bG').
\end{align*}
Therefore similarly, if we denote $\bchi_a^{(1)} = \bchi_a^{(3)} = \chi_a(\bG)$ and $\bchi_a^{(2)} = \bchi_a^{(4)} = \chi_a(\bG')$, then define $\Psi$ as
\begin{equation}\label{eq:fpoly}
    \|f(\bG)-f(\bG')\|_2^2 =: \Psi(\bchi^{(1)},\bchi^{(2)},\bchi^{(3)},\bchi^{(4)}) = \Phi(\bchi^{(1)},\bchi^{(3)})+\Phi(\bchi^{(2)},\bchi^{(4)})-\Phi(\bchi^{(1)},\bchi^{(4)})-\Phi(\bchi^{(2)},\bchi^{(3)}),
\end{equation}
where $\Psi(\cdot)$ is a multi-linear polynomial on $4\binom{n}{2}$ variables $\bchi^{(1)}$, $\bchi^{(2)}$, $\bchi^{(3)}$, and $\bchi^{(4)}$. Let $\bZ^{(1)}$, $\bZ^{(2)}$, $\bZ^{(3)}$, and $\bZ^{(4)}$ be length $\binom{n}{2}$ mean-$0$ jointly Gaussian vectors such that they are independent of the $\bchi$'s and has the same covariance structure. For simplicity, write $\vec\bchi = (\bchi^{(1)}$, $\bchi^{(2)}$, $\bchi^{(3)}$, $\bchi^{(4)})$ and $\vec \bZ = (\bZ^{(1)}$, $\bZ^{(2)}$, $\bZ^{(3)}$, $\bZ^{(4)})$.

\begin{lemma}\label{lem:hinfluence}
    Let $h:\mathbb{R} \to \mathbb{R}$ be a bounded function with bounded first three derivatives. Then
    \begin{align*}
        |\E h(\Psi(\vec \bchi))- \E h(\Psi(\vec \bZ))|\leq O\left( \frac{\|h'''(x)\|_{\infty}}{n\sqrt{q}} \right).
    \end{align*}
\end{lemma}

\begin{proof}[Proof of \pref{lem:hinfluence}]
For $\mathsf{s},\mathsf{t} \in [0,1]$, $k \in \{1,2,3,4\}$, and any edge $a$, we define $\bW^{(k),a}_{\mathsf{s},\mathsf{t}}$ to be a length $\binom{n}{2}$ vector such that for any edge $e$,
\begin{align*}
    \left(
    \bW^{(k),a}_{\mathsf{s},\mathsf{t}}\right)_{e} = \mathsf{s}\sqrt{\mathsf{t}}\bchi^{(k)}_e \Ind(e=a)+\sqrt{\mathsf{t}}\bchi^{(k)}_e\Ind(e\neq a) + \sqrt{1-\mathsf{t}}\bZ_e.
\end{align*}
Write $\vec \bW^{a}_{\mathsf{s},\mathsf{t}} = (\bW^{(1),a}_{\mathsf{s},\mathsf{t}},\bW^{(2),a}_{\mathsf{s},\mathsf{t}},\bW^{(3),a}_{\mathsf{s},\mathsf{t}},\bW^{(4),a}_{\mathsf{s},\mathsf{t}})$.
By a result of \cite{caravenna2023critical}, it suffices to show that the expected fourth moment of gradients of functions in $\bW$ are bounded.
\begin{lemma}[Lemma A.4 in \cite{caravenna2023critical}]\label{lem:influence}
    Let $h:\mathbb{R} \to \mathbb{R}$ be a bounded function with bounded first three derivatives. Then there exists an absolute constant $C$ such that 
    \begin{align*}
        |\E h(\Psi(\vec \bchi))- \E h(\Psi(\vec \bZ)))|\leq C\|h'''(x)\|_{\infty} \sup_a\E[|\chi_a(\bG)|^3] \sum_{a} \sum_{k=1}^4 \sup_{\mathsf{s},\mathsf{t} \in [0,1]} \E \left[ |\partial_{a,k} \Psi(\vec \bW^{a}_{\mathsf{s},\mathsf{t}})|^3 \right],
    \end{align*}
here if we write $\Psi(\cdot)$ as a function on $x^{(1)}, x^{(2)}, x^{(3)}, x^{(4)}$ variables, then $\partial_{a,k}$ denotes the partial derivative with respect to $x^{(k)}_a$.
\end{lemma}
The term $\E[|\chi_a(\bG)|^3] = \frac{1-2q}{\sqrt{q(1-q)}}$ by direct computation. It then remains to figure out the last expectation term. By H\"older's inequality and  Equation \eqref{eq:fpoly}, we know that for some absolute constant $C$,
\begin{equation}\label{eq:holder}
    \sup_{\mathsf{s},\mathsf{t} \in [0,1]} \E \left[ |\partial_{a,k} \Psi(\vec \bW^{a}_{\mathsf{s},\mathsf{t}})|^3 \right] \leq  \sup_{\mathsf{s},\mathsf{t} \in [0,1]} \E \left[ |\partial_{a,k} \Psi(\vec \bW^{a}_{\mathsf{s},\mathsf{t}})|^4 \right]^{3/4} \leq C\sum_{i\in \{1,2\}} \sum_{j\in \{3,4\}} \sup_{\mathsf{s},\mathsf{t} \in [0,1]} \E \left[ |\partial_{a,k} \Phi( \bW^{(i),a}_{\mathsf{s},\mathsf{t}}, \bW^{(j),a}_{\mathsf{s},\mathsf{t}})|^4 \right]^{3/4}.
\end{equation}
Now we bound the fourth moment. Note that the partial derivative $\partial_{a,k} \Phi( \bW^{(i),a}_{\mathsf{s},\mathsf{t}}, \bW^{(j),a}_{\mathsf{s},\mathsf{t}})$ is only non-zero when $k = i$ or $j$ and if so, $\partial_{a,i}\Phi( \bW^{(i),a}_{\mathsf{s},\mathsf{t}}, \bW^{(j),a}_{\mathsf{s},\mathsf{t}}) = \partial_{a,j}\Phi( \bW^{(i),a}_{\mathsf{s},\mathsf{t}}, \bW^{(j),a}_{\mathsf{s},\mathsf{t}})$, so we will drop the $k$ in the subscript and only use $\partial_a$. 
To simplify notation, we use $\bchi$ and $\tilde\bchi$ to denote $\bchi^{(i)}$ and $\bchi^{(j)}$, and similarly for $\bZ$ and $\bW$. So when $(i,j) = (1,3)$ or $(2,4)$, the correlation between $\bchi$ and $\tilde\bchi$ is $1$, and when $(i,j) = (1,4)$ or $(2,3)$, the correlation may be less than one (but non-negative). 
We will give a bound that applies to any correlation value. 
Recall that by definition,
\begin{align*}
\Phi(\bW^{a}_{\mathsf{s},\mathsf{t}},\tilde\bW^{a}_{\mathsf{s},\mathsf{t}}) & = \sum_{i,j} \sum_{H_1,H_2 \in \calH} \hat f_{H_1} \hat f_{H_2}\sum_{\ell_1 \in \calL_{H_1,i,j}, \ell_2 \in \calL_{H_2,i,j}} (\bW^{a}_{\mathsf{s},\mathsf{t}})_{\ell_1(H_1)} (\tilde\bW^{a}_{\mathsf{s},\mathsf{t}})_{\ell_2(H_2)}.
\end{align*}
Thus 
\begin{align*}
\partial_a\Phi(\bW^{a}_{\mathsf{s},\mathsf{t}},\tilde\bW^{a}_{\mathsf{s},\mathsf{t}})
     = \sum_{i,j} \sum_{H_1,H_2 \in \calH} \hat f_{H_1} \hat f_{H_2}\sum_{\substack{\ell_1 \in \calL_{H_1,i,j}, \ell_2 \in \calL_{H_2,i,j}\\ a\in \ell_1(H_1)}} (\bW^{a}_{\mathsf{s},\mathsf{t}})_{\ell_1(H_1)\backslash a} (\tilde\bW^{a}_{\mathsf{s},\mathsf{t}})_{\ell_2(H_2)},
\end{align*}
where each term $(\bW^{a}_{\mathsf{s},\mathsf{t}})_{\ell_1(H_1)\backslash a} (\tilde\bW^{a}_{\mathsf{s},\mathsf{t}})_{\ell_2(H_2)}$ equals
\begin{equation}\label{eq:partial}
    \prod_{b_1 \in \ell_1(H_1)\backslash a}(\sqrt{\mathsf{t}}\bchi_{b_1}+\sqrt{1-\mathsf{t}}\bZ_{b_1})  \prod_{b_2 \in \ell_2(H_2)}(\mathsf{s}\sqrt{\mathsf{t}}\tilde\bchi_{b_2} \Ind({b_2}=a)+\sqrt{\mathsf{t}}\tilde\bchi_{b_2}\Ind({b_2}\neq a) + \sqrt{1-\mathsf{t}}\tilde\bZ_{b_2}).
\end{equation}
We compute
\begin{align*}
    \E\partial_a \Phi(\bW^{a}_{\mathsf{s},\mathsf{t}},\tilde\bW^{a}_{\mathsf{s},\mathsf{t}})^4
    &= \E\left(\sum_{i,j} \sum_{H_1,H_2 \in \calH} \hat f_{H_1} \hat f_{H_2}\sum_{\substack{\ell_1 \in \calL_{H_1,i,j}, \ell_2 \in \calL_{H_2,i,j}\\ a\in \ell_1(H_1)}} (\bW^{a}_{\mathsf{s},\mathsf{t}})_{\ell_1(H_1)\backslash a} (\tilde\bW^{a}_{\mathsf{s},\mathsf{t}})_{\ell_2(H_2)} \right)^4\\
    & = \sum_{H_1, \cdots, H_8 \in \calH} \prod_{k=1}^8 \hat f_{H_k} \sum_{\substack{i_1, i_2, i_3, i_4\\j_1, j_2, j_3, j_4}} \sum_{\substack{\ell_1 \in \calL_{H_1,i_1,j_1}, \ell_2 \in \calL_{H_2,i_1,j_1}, a\in \ell_1(H_1)\\ \ell_3 \in \calL_{H_3,i_2,j_2}, \ell_4 \in \calL_{H_4,i_2,j_2} , a\in \ell_3(H_3)\\ \ell_5 \in \calL_{H_5,i_3,j_3}, \ell_6 \in \calL_{H_6,i_3,j_3} , a\in \ell_5(H_5)\\ \ell_7 \in \calL_{H_7,i_4,j_4}, \ell_8 \in \calL_{H_8,i_4,j_4} , a\in \ell_7(H_7) } } \E \left[ \prod_{k=1}^4 (\bW^{a}_{\mathsf{s},\mathsf{t}})_{\ell_{2k-1}(H_{2k-1})\backslash a} (\tilde\bW^{a}_{\mathsf{s},\mathsf{t}})_{\ell_{2k}(H_{2k})} \right].
\end{align*}
To further simplify, we use equation \eqref{eq:partial} to expand the product. 
Note that since we have assumed the correlation is non-negative, all the terms have non-negative expectation, therefore we can bound $\mathsf{s}\sqrt{\mathsf{t}}$ by $\sqrt{\mathsf{t}}$ and have that
\begin{align*}
    \E \left[ \prod_{k=1}^4 (\bW^{a}_{\mathsf{s},\mathsf{t}})_{\ell_{2k-1}(H_{2k-1})\backslash a} (\tilde\bW^{a}_{\mathsf{s},\mathsf{t}})_{\ell_{2k}(H_{2k})} \right] \leq \E \left[ \prod_{k=1}^4 (\sqrt{\mathsf{t}}\bchi+\sqrt{1-\mathsf{t}}\bZ)_{\ell_{2k-1}(H_{2k-1})\backslash a} (\sqrt{\mathsf{t}}\tilde\bchi+\sqrt{1-\mathsf{t}}\tilde\bZ)_{\ell_{2k}(H_{2k})} \right].
\end{align*}
Next we control the right hand side.
\begin{lemma}\label{lem:momentbound}
    For any $k_1,k_2\geq 0$ such that $2\leq k= k_1+k_2 \leq 8$, and any edge $b$,
    \begin{align*}
        \E\left[ (\sqrt{\sft}\bchi_b + \sqrt{1-\sft}Z_b)^{k_1} (\sqrt{\sft}\tilde\bchi_b + \sqrt{1-\sft}\tilde\bZ_b)^{k_2} \right] \leq 2 (q/2)^{-(k-2)/2}.
    \end{align*}
\end{lemma}
\begin{proof}[Proof of \pref{lem:momentbound}]
    By direct computation, $\E\left[ \bchi_b \right] = 0$ and for any $2 \leq k \leq 8$,
    \begin{align*}
        \E\left[ \bchi_b^k \right] = \left(\frac{1-q}{\sqrt{q(1-q)}} \right)^k q + \left(\frac{-q}{\sqrt{q(1-q)}} \right)^k (1-q) \leq \frac{(1-q)^{k-1}+q^{k-1}}{(q(1-q))^{(k-2)/2}} \leq q^{-(k-2)/2},
    \end{align*}
    where the last inequality is equality for $k=2$ and when $k\geq 3$, since $q = o(1)$, $(1-q)^{k-1}+q^{k-1} = 1-q(k-1) + O(q^2) < 1-q (k-2)/2 \leq (1-q)^{(k-2)/2}$. For correlated terms, we let $\bchi$ and $\tilde\bchi$ be $\rho$ correlated. Then we write $p_1 = \rho q(1-q)+q^2$, $p_2 = \rho q(1-q)+(1-q^2)$, and $p_3 = (1-\rho)q(1-q)$, and denote $q_1 = (1-q)/\sqrt{q(1-q)}$ and $q_2 = -q/\sqrt{q(1-q)}$. For $k_1,k_2\geq 0$ such that $2\leq k= k_1+k_2 \leq 8$,
    \begin{align*}
        &\E\left[ \bchi_b^{k_1} \tilde\bchi_b^{k_2} \right] = q_1^k p_1 + q_2^k p_2 + q_1^{k_1} q_2^{k_2} p_3 + q_2^{k_1}  q_1^{k_2} p_3 \leq q_1^k p_1 + |q_2|^k p_2 + (q_1^k+|q_2|^k) p_3 \\
        &= \frac{(1-q)^{k-1}+q^{k-1}}{(q(1-q))^{(k-2)/2}} \leq q^{-(k-2)/2}.
    \end{align*}
    For the Gaussian term $\bZ_b$, $\E[\bZ_b] = 0$, $\E[\bZ_b^2] = 1$ and for any $3\leq k \leq 8$, $\E[\bZ_b^k] \leq 105 \leq q^{-(k-2)/2}$. And by Wick's theorem, this is also true for correlated terms that $\E[\bZ_b^{k_1} \tilde\bZ_b^{k_2}] \leq \E[\bZ_b^{k_1+k_2}] \leq q^{-(k-2)/2}$ for $2\leq k= k_1+k_2 \leq 8$. Now we claim that for any $0\leq c_1\leq k_1$ and $0\leq c_2\leq k_2$, with $2\leq k = k_1+k_2 \leq 8$, we have
    \begin{align*}
        \E \left[ \bchi_b^{c_1} \tilde\bchi_b^{c_2} \right] \E \left[ \bZ_b^{k_1-c_1} \tilde\bZ_b^{k_2-c_2} \right] \leq q^{-(k-2)/2}.
    \end{align*}
    If $c_1+c_2 \leq 2$ or $k_1+k_2-c_1-c_2 \leq 2$, then the above holds by considering different cases and combining the previous estimates for $\E \left[ \bchi_b^{c_1} \tilde\bchi_b^{c_2} \right]$ and $\E \left[ \bZ_b^{k_1-c_1} \tilde\bZ_b^{k_2-c_2} \right]$. Else, 
    \begin{align*}
        \E \left[ \bchi_b^{c_1} \tilde\bchi_b^{c_2} \right] \E \left[ \bZ_b^{k_1-c_1} \tilde\bZ_b^{k_2-c_2} \right] \leq q^{-(c_1+c_2-2)/2} q^{-(k_1+k_2-c_1-c_2-2)/2} \leq q^{-(k-4)/2} \leq q^{-(k-2)/2}.
    \end{align*}
    Therefore,
    \begin{align*}
        &\E\left[ (\sqrt{\sft}\bchi_b + \sqrt{1-\sft}\bZ_b)^{k_1} (\sqrt{\sft}\tilde\bchi_b + \sqrt{1-\sft}\tilde\bZ_b)^{k_2} \right] \\
        &= \sum_{c_1 = 0}^{k_1}\binom{k_1}{c_1} \sqrt{\sft}^{c_1} \sqrt{1-\sft}^{k_1-c_1} \sum_{c_2 = 0}^{k_2}\binom{k_2}{c_2} \sqrt{\sft}^{c_2} \sqrt{1-\sft}^{k_2-c_2} \E \left[ \bchi_b^{c_1} \tilde\bchi_b^{c_2} \right] \E \left[ \bZ_b^{k_1-c_1} \tilde\bZ_b^{k_2-c_2} \right]\\
        & \leq \sum_{c_1 = 0}^{k_1}\binom{k_1}{c_1} \sqrt{\sft}^{c_1} \sqrt{1-\sft}^{k_1-c_1} \sum_{c_2 = 0}^{k_2}\binom{k_2}{c_2} \sqrt{\sft}^{c_2} \sqrt{1-\sft}^{k_2-c_2} q^{-(k-2)/2} \\
        & \leq (\sqrt{\sft}+\sqrt{1-\sft})^k q^{-(k-2)/2} \leq 2 (q/2)^{-(k-2)/2},
    \end{align*}
    which completes the proof of \pref{lem:momentbound}.
\end{proof}
Finally, we will require the following combinatorial lemma for bounding the number of summands for each collection of subgraphs.
\begin{lemma}\label{lem:combo-graph}
   Suppose $H_1,\ldots,H_m \in \calH$.
   Let $H_\sep=\bigcup_{i=1}^m H_i$ have $\nu_{\sep}$ vertices and $\mu_\sep$ edges (the special vertices $s,t$, if present in several of the $H_i$, are counted only with multiplicity one).
   For each $i \in [m]$, let $\ell_i:V(H_i) \hookrightarrow \N$ be a one-to-one labeling of the vertices, satisfying $\ell_i(s) = 1$ and $\ell_i(t) = 2$. 
   Call $H_{\combo}$ the labeled simple graph $\bigcup_{i=1}^m \ell_i(H_i)$, with $\nu_{\combo}$ vertices, and $\mu_\combo$ edges. 
   Then if every edge in $H_{\combo}$ is covered at least twice,
   \[
        \nu_\combo \le \frac{1}{2}\left(\nu_\sep  +  \Ind[s \in V(H_\sep)] + \Ind[t \in V(H_\sep)] - (\mu_\sep - 2\mu_\combo)\right).
   \]
\end{lemma}
\begin{proof}
    If a graph has $\nu$ vertices, $\mu$ edges, $\kappa$ components, then
    \[
    \nu = \kappa + \mu - \gamma,
    \]
    where $\gamma$ is the number of edges one must remove to obtain the spanning forest of the graph.
    We will apply this identity to $H_\combo$, then use the fact that each $H_i$ is of uniform density at most $1$ to relate $\kappa_\combo,\mu_\combo,\gamma_\combo$ to the related quantities of the $H_i$.
    
    First, some observations about edges. 
    Say that $H_\combo$ has $\tilde{\mu}_\ell$ edges covered by $\ell$ edges in $H_\sep$ for each $\ell \in \{2,\ldots,m\}$.
    Since each edge in $H_\sep$ maps to an edge in $H_\combo$, and each edge in $H_\combo$ is covered twice, $\tilde \mu_1 = 0$, so
    \[
    \mu_\sep = \sum_{\ell=2}^m \ell\tilde{\mu}_\ell \quad \text{ and }\quad \mu_\combo = \sum_{\ell=2}^m \tilde{\mu}_\ell.
    \]

    Now, some observations about components.
    In $H_\sep$, every component comes from some distinct $H_i$, except for up to two components $C_s,C_t$ which contain the special vertices $s$ and $t$; these may be composed of several $H_i$.
    Since the $H_i \in \calH$, and since all graphs in $\calH$ have the property that the components including $s$ or $t$ must be disjoint and trees, the components $C_s,C_t$ of $H_\sep$ are also trees (or possibly are just the empty graph).
    We will say that in $H_\combo$, for each $\ell \in \N$, there are $\tilde{\kappa}_\ell$ components which result from combining $\ell$ components from the $H_\sep$.
    Since by assumption each edge in $H_\combo$ is covered at least twice, and each $H_i$ has all-distinct edges, $\tilde{\kappa}_1 \le \Ind[C_s\neq \emptyset] + \Ind[C_t \neq \emptyset]$, as any other component in $H_\combo$ has to combine at least two components to be double-covered.
    Thus,
    \[
    \kappa_\combo = \sum_{\ell=1}^\infty \tilde{\kappa}_\ell, \quad \tilde\kappa_1 \le \Ind[C_s \neq \emptyset] + \Ind[C_t \neq \emptyset], \quad \text{ and }\quad 
        \kappa_\sep = \sum_{\ell=2}^\infty \ell \tilde\kappa_\ell.
    \]
    
    Finally, we make some observations about cycles.
    Let $\xi_\ell$ be the number of components in $H_\combo$ which contain the image of exactly $\ell$ cycles from $H_\sep$.
    Then because each connected component in $H_\sep$ contains at most once cycle, $\gamma_\sep = \sum_{\ell=1}^m \ell \xi_\ell$.
    We also have that 
    \[
    \gamma_\combo \ge \sum_{\ell=1}^m \xi_\ell,
    \]
    as one must remove at least one edge from each cycle in $H_\sep$ to get a forest, and if there are only $\ell$ cycles which map together to a component of $H_\combo$, then at least one edge must be removed.

    Now, we have that
    \begin{align*}
    \nu_\combo 
    &= \kappa_\combo + \mu_\combo - \gamma_\combo\\
    &\le \sum_{\ell=1}^\infty \tilde{\kappa}_\ell + \sum_{\ell=2}^m \tilde\mu_\ell - \sum_{\ell=1}^m \xi_\ell\\
    &= \sum_{\ell=1}^\infty \tilde{\kappa}_\ell + \sum_{\ell=2}^m \tilde\mu_\ell - \sum_{\ell=1}^m \xi_\ell + \frac{1}{2}\left(\kappa_\sep - \sum_{\ell=2}^\infty \ell \tilde \kappa_\ell \right) + \frac{1}{2}\left(\mu_\sep - \sum_{\ell=2}^m \ell \tilde{\mu}_\ell \right) + \frac{1}{2}\left(\sum_{\ell=1}^m \ell \xi_\ell - \gamma_\sep \right) \\
    &= \frac{1}{2}(\kappa_\sep + \mu_\sep - \gamma_\sep) + \frac{1}{2}\tilde\kappa_1 - \sum_{\ell=3}^\infty \frac{\ell-2}{2} \tilde{\kappa}_\ell - \sum_{\ell=3}^m \frac{\ell-2}{2} \tilde{\mu}_\ell + \sum_{\ell=1}^m \frac{\ell-2}{2} \xi_\ell\\
    &= \frac{1}{2}\nu_\sep + \frac{1}{2}\tilde{\kappa}_1 - \sum_{\ell=3}^m \frac{\ell-2}{2} \tilde{\mu}_\ell + \sum_{\ell=3}^\infty \frac{\ell-2}{2}\left(\xi_\ell - \tilde{\kappa}_\ell\right)\\
    &\le \frac{1}{2}\nu_\sep + \frac{1}{2}(\Ind[C_s \neq \emptyset] + \Ind[C_t \neq \emptyset]) - \frac{1}{2}(\mu_\sep - 2\mu_\combo) + \sum_{\ell=3}^\infty \frac{\ell-2}{2}\left(\xi_\ell - \tilde{\kappa}_\ell\right),
    \end{align*}
    We'll argue that the rightmost sum is at most zero, which will conclude the proof.
    Indeed, note that in $H_\sep$, each connected component contained at most one cycle.
    Hence, any connected component of $H_\combo$ with $\ell$ cycles must be in the image of $\ell$ or more components from $H_\sep$.
    Hence the positive contribution of each component to $\xi_\ell$ is negated by its negative contribution to $\tilde\kappa_{\ell'}$ for some $\ell' \ge \ell$, completing the proof.
\end{proof}

Now back to the computation of the fourth moment. 
Call $U = \cup_{i=1}^8 H_i$ to be the union graph of the unlabeled shapes, in which only the $s,t$ vertices are identified if present. 
Let $\calC(H_1,\cdots,H_8)$ be the set of all possible combination graphs of the $H_i$, such that the edge $a$ must be covered in all $H_1, H_3, H_5, H_7$ and all edges are covered at least twice in $U_a = \left(\bigcup_{i\in \{1,3,5,7\} } H_i\backslash a \right) \cup \left(\bigcup_{i\in \{2,4,6,8\} } H_i \right)$. 
Let $C$ be any graph in $\calC(H_1,\ldots,H_8)$, and write $C_a$ as the corresponding combination graph of $U_a$. 
The special vertices $\{i_{2k-1},i_{2k}\}$ of $H_{2k-1},H_{2k}$ are identified if present, and the special vertices $\{s,t\}$ of all eight graphs are identified if present. 
Let $v_C$ be the number of vertices in $C$ excluding the special vertices $\{s,t\}$, and $s_C$ be the number of $\{s,t\}$ vertices in $C$, and let $e_C$ be the number of edges in $C$, and define $v,s,e$ similarly for $C_a$ and the union graphs $U$ and $U_a$. We will show that 
\begin{equation}\label{eq:edgevertexrelation}
    v_C\leq \frac{1}{2}v_U - \frac{1}{2}e_U + e_{C_a}.
\end{equation}
Indeed, for $C \in \calC(H_1,\cdots,H_8)$, we consider eight graphs $H_1\backslash a, H_2, H_3 \backslash a, H_4, H_5, H_6, H_7, H_8$, and define $U'$ as their union and $C'$ as their combination. Note that for $H_1\backslash a$ and $H_3\backslash a$, we only delete the edge $a$ but not vertices. Then $C'$ and $U'$ satisfy the conditions of \pref{lem:combo-graph}, and we have that
\[
v_{C'} + s_{C'} \le \frac{1}{2}(v_{U'} + s_{U'}) + \frac{1}{2} s_{U'} -\frac{1}{2}(e_{U'} - 2e_{C'}) \,\implies\, v_{C'} \le \frac{1}{2}v_{U'} - \frac{1}{2}e_{U'}+e_{C'}.
\]
Therefore equation \eqref{eq:edgevertexrelation} follows by noticing that $v_C = v_{C'}$, $v_U = v_{U'}$, $e_U = e_{U'}+2$, $e_{C_a} \geq e_{C'}-1$. Let $a_1$ and $a_2$ be the two vertices of the edge $a$. Then by \pref{lem:momentbound} and equation \eqref{eq:edgevertexrelation}, if $\{a_1,a_2\} \cap \{s,t\} = \emptyset$, then
\begin{align*}
    \E\partial_a \Phi(\bW^{a}_{s,t},\tilde\bW^{a}_{s,t})^4  &= \sum_{H_1, \cdots, H_8 \in \calH} \prod_{k=1}^8 \hat f_{H_k} \sum_{\substack{i_1, i_2, i_3, i_4\\j_1, j_2, j_3, j_4}} \sum_{\substack{\ell_1 \in \calL_{H_1,i_1,j_1}, \ell_2 \in \calL_{H_2,i_1,j_1}, a\in \ell_1(H_1)\\ \ell_3 \in \calL_{H_3,i_2,j_2}, \ell_4 \in \calL_{H_4,i_2,j_2} , a\in \ell_3(H_3)\\ \ell_5 \in \calL_{H_5,i_3,j_3}, \ell_6 \in \calL_{H_6,i_3,j_3} , a\in \ell_5(H_5)\\ \ell_7 \in \calL_{H_7,i_4,j_4}, \ell_8 \in \calL_{H_8,i_4,j_4} , a\in \ell_7(H_7) } } \E \left[ \prod_{k=1}^4 (\bW^{a}_{s,t})_{\ell_{2k-1}(H_{2k-1})\backslash a} (\tilde\bW^{a}_{s,t})_{\ell_{2k}(H_{2k})} \right]\\
    &\leq \sum_{H_1, \cdots, H_8 \in \calH} \prod_{k=1}^8 \hat f_{H_k} \sum_{C\in \calC(H_1, \cdots, H_8)} n^{v_{C}-2} \cdot 2(q/2)^{-\frac{1}{2}(e_{U_a} - 2e_{C_a})}\\
    &= \sum_{H_1, \cdots, H_8 \in \calH} \prod_{k=1}^8 O \left( n^{-\frac{1}{2}(v_{H_k}+u_{H_k})}\right) \sum_{C\in \calC(H_1, \cdots, H_8)} n^{\frac{1}{2}(e_{U_a} - 2e_{C_a}) + v_{C}-2 } (nq/2)^{-\frac{1}{2}(e_{U_a} - 2e_{C_a})}\\
    &\leq O \left( n^{-\frac{1}{2}\sum_{k=1}^8 (v_{H_k}+u_{H_k}) +\frac{1}{2}(e_{U_a} - 2e_{C_a}) + v_{C}-2} \right) \leq O(n^{-4}),
\end{align*}
where in line $3$ we have used our bound on the $\hat f_H$, and the final line uses that $-\frac{1}{2}\sum_{k=1}^8 (v_{H_k}+u_{H_k}) +\frac{1}{2}(e_{U_a} - 2e_{C_a}) + v_{C} = -\frac{1}{2} v_U + \frac{1}{2} (e_U - 4) - e_{C_a} + v_{C} \leq -2 $. 

If $|\{a_1,a_2\} \cap \{s,t\}| = 1$, then we can replace in the second line $n^{v_C-2}$ by $n^{v_C-1}$, and get $O(n^{-3})$. Similarly, if $\{a_1,a_2\} = \{s,t\}$, then we can replace in the second line $n^{v_C-2}$ by $n^{v_C}$, and get $O(n^{-2})$. Combining with \pref{lem:influence} and Equation \eqref{eq:holder}, we have that
\begin{align*}
    |\E h(\Psi(\vec \bchi))- \E h(\Psi(\vec \bZ))| &\leq \frac{\|h'''(x)\|_{\infty}}{\sqrt{q}} \left(  n^2 O(n^{-4 \cdot \frac{3}{4}}) + nO(n^{-3 \cdot \frac{3}{4}}) + 1 O(n^{-2 \cdot \frac{3}{4}})\right)
    = O\left( \frac{\|h'''(x)\|_{\infty}}{n\sqrt{q}} \right),
\end{align*}
which completes the proof of \pref{lem:hinfluence}.
\end{proof}

We next prove \pref{thm:smooth}.
\begin{proof}[Proof of \pref{thm:smooth}]
Recall that 
\[
f_{ij}(\bG) = \sum_{H \in \calH} \hat f_{H} \sum_{\ell \in \calL_{H,i,j}} \bchi_{\ell(H)}(\bG).
\]
We do a change of variables and define $g_{ij}(\bchi) = f_{ij}(\bG)$. Then $g$ is also a vector-valued polynomial of degree $D$. And
$\Psi(\vec \bchi) = \|g(\bchi)-g(\bchi')\|_2^2$ and $\Psi(\vec \bZ) = \|g(\bZ)-g(\bZ')\|_2^2$, where $\bchi$ and $\bchi'$ are $\rho$-correlated and the same for $\bZ$ and $\bZ'$. By \pref{lem:hinfluence}, we know that
\begin{equation}\label{eq:smoothed-diff}
    |\E h(\|g(\bchi)-g(\bchi')\|_2^2)- \E h(\|g(\bZ)-g(\bZ')\|_2^2)|\leq O\left( \frac{\|h'''(x)\|_{\infty}}{n\sqrt{q}} \right),
\end{equation}
for bounded $h$ with bounded first three derivatives. We define a smooth function $\phi$
    \begin{equation*}
        \phi(x) = \begin{cases}
            \exp(-\frac{1}{x}), \quad &\text{if } x>0,\\
            0, \quad &\text{if } x\leq 0.
        \end{cases}
    \end{equation*}
    For any $\delta>0$, we define a smooth transition function
    \begin{equation*}
        h_\delta(x) = \frac{\phi(x/\delta)}{\phi(1-x/\delta) + \phi(x/\delta)}.
    \end{equation*}
    Note that $h_\delta$ is smooth, increasing, $h_\delta(x) = 0$ when $x \leq 0$, and $h_\delta(x) = 1$ when $x \geq \delta$. Moreover, there is a universal constant $K$ such that $|h'''_\delta| \leq K/\delta^3$. By Equation \eqref{eq:smoothed-diff}, this implies that
    \begin{align*}
        &\Pr\left[ \|f(\bG)-f(\bG')\|_2^2 \geq \gamma \right] = \Pr\left[ \|g(\bchi)-g(\bchi')\|_2^2 \geq \gamma \right]
        \leq \E \left[ h_\delta \left( \|g(\bchi)-g(\bchi')\|_2^2 -\gamma + \delta \right) \right]\\
        &\leq O\left( \frac{1}{n\sqrt{q} \delta^3} \right) + \E \left[ h_\delta \left( \|g(\bZ)-g(\bZ')\|_2^2 -\gamma + \delta \right) \right]
        \leq O\left( \frac{1}{n\sqrt{q} \delta^3} \right) + \Pr \left[ \|g(\bZ)-g(\bZ')\|_2^2 \geq \gamma-\delta \right].
    \end{align*}
    To control the last probability term, we cite the following result
    \begin{lemma}[Theorem 3.1 in \cite{GJW24}]
        Let $0 \leq \rho \leq 1$. Let $\bZ, \bZ'$ be a pair of standard Gaussian random vectors on $\mathbb{R}^{d}$ that are $\rho$-correlated. Let $g: \mathbb{R}^d \rightarrow \mathbb{R}^k$ be a (deterministic) polynomial of degree at most $D$ with $\E\|g(\bZ)\|_2^2=1$. For any $t\geq(6 e)^D$,
        \[
        \Pr\left(\|g(\bZ)-g(\bZ')\|_2^2 \geq 2 t\left(1-\rho^D\right)\right) \leq \exp \left(-\frac{D}{3 e} t^{1 / D}\right).
        \]
    \end{lemma}
    For $\rho \geq 1-\frac{1}{T}$, we take $t = \frac{\gamma}{3(1-\rho^D)}$ and $\delta = \gamma/3$. Then 
    \begin{align*}
        \Pr\left[ \|f(\bG)-f(\bG')\|_2^2 \geq \gamma \right] &\leq O\left( \frac{1}{n\sqrt{q} \gamma^3} \right) + \exp\left(-\frac{D}{3 e} \left( \frac{\gamma}{3(1-\rho^D)} \right)^{1 / D}\right)\\
        & \leq O\left( \frac{1}{n\sqrt{q} \gamma^3} \right) + \exp\left(-\frac{D}{3 e} \left( \frac{\gamma T}{3D} \right)^{1 / D}\right).
    \end{align*}
This, in combination with our normalization of $f$, completes the proof of \pref{thm:smooth}.
\end{proof}

\ifnum\anon=0
\subsection*{Acknowledgements}
We thank David Gamarnik, Andrea Montanari, Guy Bresler, and Sam Hopkins for helpful conversations.
We also thank anonymous reviewers for helpful comments on the manuscript.
\fi
\bibliographystyle{alpha}
\bibliography{refs}

\newcommand{\etalchar}[1]{$^{#1}$}
\begin{thebibliography}{BMVT78}

\bibitem[AR02]{AR02}
Vikraman Arvind and Venkatesh Raman.
\newblock Approximation algorithms for some parameterized counting problems.
\newblock In {\em Algorithms and Computation: 13th International Symposium, {ISAAC} 2002 Vancouver, BC, Canada, November 21--23, 2002 Proceedings 13}, pages 453--464. Springer, 2002.

\bibitem[ART06]{ART06}
Dimitris Achlioptas and Federico Ricci-Tersenghi.
\newblock On the solution-space geometry of random constraint satisfaction problems.
\newblock In {\em Proceedings of the thirty-eighth annual {ACM} symposium on Theory of computing}, pages 130--139, 2006.

\bibitem[AS16]{alon2016probabilistic}
Noga Alon and Joel~H Spencer.
\newblock {\em The probabilistic method}.
\newblock John Wiley \& Sons, 2016.

\bibitem[AYZ95]{AYZ95}
Noga Alon, Raphael Yuster, and Uri Zwick.
\newblock Color-coding.
\newblock {\em J. {ACM}}, 42(4):844–856, jul 1995.

\bibitem[BBH{\etalchar{+}}21]{brennan2021statistical}
Matthew~S Brennan, Guy Bresler, Sam Hopkins, Jerry Li, and Tselil Schramm.
\newblock Statistical query algorithms and low degree tests are almost equivalent.
\newblock In {\em Conference on Learning Theory}, pages 774--774. {PMLR}, 2021.

\bibitem[BGL14]{BGL}
Dominique Bakry, Ivan Gentil, and Michel Ledoux.
\newblock {\em Analysis and geometry of Markov diffusion operators}, volume 103.
\newblock Springer, 2014.

\bibitem[BGRS04]{BGRS04}
B{\'e}la Bollob{\'a}s*, David Gamarnik, Oliver Riordan, and Benny Sudakov.
\newblock On the value of a random minimum weight steiner tree.
\newblock {\em Combinatorica}, 24(2):187--207, 2004.

\bibitem[BH22]{BH22}
Guy Bresler and Brice Huang.
\newblock The algorithmic phase transition of random k-sat for low degree polynomials.
\newblock In {\em 2021 {IEEE} 62nd Annual Symposium on Foundations of Computer Science ({FOCS})}, pages 298--309. {IEEE}, 2022.

\bibitem[BLM15]{BLM15}
Mohsen Bayati, Marc Lelarge, and Andrea Montanari.
\newblock Universality in polytope phase transitions and message passing algorithms.
\newblock {\em Annals of applied probability: an official journal of the Institute of Mathematical Statistics}, 25(2):753--822, 2015.

\bibitem[BMVT78]{BMT78}
{ER} Berlekamp, {RJ} Mceliece, and {HCA} Van~Tilborg.
\newblock On the inherent intractability of certain coding problems.
\newblock {\em {IEEE} Transactions on Information Theory}, 1978.

\bibitem[BvdH12]{BvdH}
Shankar Bhamidi and Remco van~der Hofstad.
\newblock Weak disorder asymptotics in the stochastic mean-field model of distance.
\newblock {\em The Annals of Applied Probability}, pages 29--69, 2012.

\bibitem[CSZ23]{caravenna2023critical}
Francesco Caravenna, Rongfeng Sun, and Nikos Zygouras.
\newblock The critical 2d stochastic heat flow.
\newblock {\em Inventiones mathematicae}, 233(1):325--460, 2023.

\bibitem[EAMS22]{el2022sampling}
Ahmed El~Alaoui, Andrea Montanari, and Mark Sellke.
\newblock Sampling from the {S}herrington-{K}irkpatrick {G}ibbs measure via algorithmic stochastic localization.
\newblock In {\em 2022 {IEEE} 63rd Annual Symposium on Foundations of Computer Science ({FOCS})}, pages 323--334. {IEEE}, 2022.

\bibitem[Gam21]{Gam21}
David Gamarnik.
\newblock The overlap gap property: A topological barrier to optimizing over random structures.
\newblock {\em Proceedings of the National Academy of Sciences}, 118(41):e2108492118, 2021.

\bibitem[GJW24]{GJW24}
David Gamarnik, Aukosh Jagannath, and Alexander~S Wein.
\newblock Hardness of random optimization problems for {B}oolean circuits, low-degree polynomials, and {L}angevin dynamics.
\newblock {\em {SIAM} Journal on Computing}, 53(1):1--46, 2024.

\bibitem[GK23]{GK23}
David Gamarnik and Eren~C K{\i}z{\i}lda{\u{g}}.
\newblock Algorithmic obstructions in the random number partitioning problem.
\newblock {\em The Annals of Applied Probability}, 33(6B):5497--5563, 2023.

\bibitem[GS17]{GS17}
David Gamarnik and Madhu Sudan.
\newblock Limits of local algorithms over sparse random graphs.
\newblock {\em Annals of probability: An official journal of the Institute of Mathematical Statistics}, 45(4):2353--2376, 2017.

\bibitem[Hop18]{hopkins2018statistical}
Samuel Hopkins.
\newblock {\em Statistical inference and the sum of squares method}.
\newblock PhD thesis, Cornell University, 2018.

\bibitem[HS22]{HS22}
Brice Huang and Mark Sellke.
\newblock Tight {L}ipschitz hardness for optimizing mean field spin glasses.
\newblock In {\em 2022 {IEEE} 63rd Annual Symposium on Foundations of Computer Science ({FOCS})}, pages 312--322. {IEEE}, 2022.

\bibitem[HS23]{HS23}
Brice Huang and Mark Sellke.
\newblock Algorithmic threshold for multi-species spherical spin glasses.
\newblock {\em arXiv preprint arXiv:2303.12172}, 2023.

\bibitem[HS25]{HS25}
Brice Huang and Mark Sellke.
\newblock Strong low degree hardness for stable local optima in spin glasses.
\newblock {\em arXiv preprint arXiv:2501.06427}, 2025.

\bibitem[IKKM12]{IKKM12}
Morteza Ibrahimi, Yashodhan Kanoria, Matt Kraning, and Andrea Montanari.
\newblock The set of solutions of random {XORSAT} formulae.
\newblock In {\em Proceedings of the twenty-third annual {ACM-SIAM} symposium on Discrete Algorithms}, pages 760--779. {SIAM}, 2012.

\bibitem[IS24]{IS24}
Misha Ivkov and Tselil Schramm.
\newblock Semidefinite programs simulate approximate message passing robustly.
\newblock In {\em Proceedings of the 56th Annual {ACM} Symposium on Theory of Computing}, pages 348--357, 2024.

\bibitem[Kun24]{kunisky2024low}
Dmitriy Kunisky.
\newblock Low coordinate degree algorithms i: Universality of computational thresholds for hypothesis testing.
\newblock {\em arXiv preprint arXiv:2403.07862}, 2024.

\bibitem[KWB19]{KWB19}
Dmitriy Kunisky, Alexander~S Wein, and Afonso~S Bandeira.
\newblock Notes on computational hardness of hypothesis testing: Predictions using the low-degree likelihood ratio.
\newblock In {\em {ISAAC} Congress (International Society for Analysis, its Applications and Computation)}, pages 1--50. Springer, 2019.

\bibitem[MMZ05]{MMZ05}
Marc M{\'e}zard, Thierry Mora, and Riccardo Zecchina.
\newblock Clustering of solutions in the random satisfiability problem.
\newblock {\em Physical Review Letters}, 94(19):197205, 2005.

\bibitem[Mon21]{Mon21}
Andrea Montanari.
\newblock Optimization of the {S}herrington--{K}irkpatrick {H}amiltonian.
\newblock {\em {SIAM} Journal on Computing}, (0):{FOCS}19--1, 2021.

\bibitem[RV17]{RV17}
Mustazee Rahman and B{\'a}lint Vir{\'a}g.
\newblock Local algorithms for independent sets are half-optimal.
\newblock {\em Annals of Probability}, 45(3), 2017.

\bibitem[Wei22]{Wein22}
Alexander~S Wein.
\newblock Optimal low-degree hardness of maximum independent set.
\newblock {\em Mathematical Statistics and Learning}, 4(3):221--251, 2022.

\end{thebibliography}

\end{document}